\documentclass[journal,12pt,onecolumn,draftclsnofoot]{IEEEtran}

\usepackage{enumitem}
\usepackage[yyyymmdd,hhmmss]{datetime}
\usepackage{amsmath, amsfonts, bm, amssymb, amsthm, comment}
\usepackage{array}
\usepackage{textcomp, mdframed}
\usepackage{stfloats}
\usepackage{verbatim}
\usepackage{graphicx, multirow}
\usepackage{cite}
\usepackage{kotex}
\usepackage[dvipsnames]{xcolor}
\usepackage{epsfig, epstopdf}
\usepackage{ulem}
\usepackage{soul}
\usepackage{algorithm, algorithmicx}
\usepackage{algpseudocode}
\usepackage[justification=centering]{caption}
\usepackage{float}
\usepackage{hyperref}
\usepackage{booktabs}
\usepackage[permil]{overpic}
\usepackage{pict2e}
\usepackage{placeins}
\usepackage{afterpage}
\usepackage{mathtools}
\usepackage{subcaption}

\newtheorem{lemma}{Lemma}
\newtheorem{proposition}{Proposition}

\theoremstyle{definition}
\newtheorem*{remark}{Remark}

\newcommand{\dg}{\dagger}

\newcommand{\KJcancel}[1]{}
\renewcommand{\KJcancel}[1]{{\color{blue}\sout{#1}}}
\newcommand{\KJcancelc}[1]{}

\newcommand{\KJcancelEQ}[1]{}
\renewcommand{\KJcancelEQ}[1]{{\color{red}#1}}

\newcommand{\JHbox}[1]{\vspace{.2in}\noindent\fbox{\begin{minipage}{\columnwidth}#1\end{minipage}}\vspace{.2in}}
\renewcommand{\JHbox}[1]{}

\renewcommand{\Delta}{\delta}

\definecolor{JHBGcolor}{HTML}{C7D7C7}

\sethlcolor{yellow}

\begin{document}
	\title{Weighted-Sum Energy Efficiency Maximization in User-Centric Uplink Cell-Free Massive MIMO}
	
	\author{Donghwi Kim,~\IEEEmembership{Graduate Student Member,~IEEE,} Liesbet Van der Perre,~\IEEEmembership{Senior Member,~IEEE}, and Wan Choi,~\IEEEmembership{Fellow,~IEEE}
		
		\thanks{D. Kim and W. Choi are with the Department of Electrical and Computer Engineering, Seoul National University (SNU), and the Institute of New Media and Communications, SNU, Seoul 08826, Korea (e-mail: yob0322@snu.ac.kr; wanchoi@snu.ac.kr). (\emph{Corresponding Author: Wan Choi})}
		\thanks{L. Van der Perre is with ESAT-WaveCore, KU Leuven, Ghent, Belgium. (e-mail: liesbet.vanderperre@kuleuven.be).}}
	
	\maketitle
	
	\begin{abstract}
		This paper introduces the weighted-sum energy efficiency (WSEE) as an advanced performance metric designed to represent the uplink energy efficiency (EE) of individual user equipment (UE) in a user-centric Cell-Free massive MIMO (CF-mMIMO) system more accurately. In a realistic user-centric CF-mMIMO context, each UE may exhibit distinct characteristics, such as maximum transmit power limits or specific minimum data rate requirements. By computing the EE of each UE independently and adjusting the weights accordingly, the system can accommodate these unique attributes, thus promoting energy-efficient operation. The uplink WSEE is formulated as a multiple-ratio fractional programming (FP) problem, representing a weighted sum of the EE of individual UEs, which depends on each UE’s transmit power and the combining vector at the central processing unit (CPU). 
		To effectively maximize WSEE, we develop optimization algorithms based on the quadratic transform (QT), which is effective for multiple-ratio FP. By applying QT sequentially to each user’s EE and the uplink SINR, the method converts the nonconvex WSEE objective into tractable subproblems and ensures stable, monotone convergence. We further introduce an approximate variant that alleviates QT’s inherent nonlinearities to accelerate convergence. Compared with global energy efficiency (GEE)–oriented baselines, the proposed algorithms yield simultaneous improvements in user power consumption and spectral efficiency, while also reducing optimization time. Overall, the framework provides a foundation for designing operational strategies tailored to specific system requirements.
		%To effectively maximize WSEE, we present optimization algorithms that utilize the Dinkelbach transform and the quadratic transform (QT). Applying the QT twice consecutively yields significant performance gains in terms of WSEE. This framework establishes a foundation for developing operational strategies tailored to specific system requirements.
	\end{abstract}
	
	\begin{IEEEkeywords}
		Cell-free massive MIMO, user-centric association, weighted-sum energy efficiency (WSEE), fractional programming, quadratic transform
	\end{IEEEkeywords}
	
	\section{Introduction\label{intro}}
	\IEEEPARstart{T}{he} growing demand for higher data transmission rates and increasing diversity of services remain major challenges in future wireless communication systems \cite{jiang2021road}, \cite{ituFrameworkOverall}. To achieve high data rates while maintaining quality of service (QoS) across numerous wireless devices, massive multiple-input multiple-output (MIMO) architectures have been standardized as a key 5G technology and widely adopted. Among various massive MIMO architectures, network or distributed MIMO utilizes a spatially distributed set of antennas to cooperatively reduce inter-cell interference in traditional cellular systems, allowing for higher data rates \cite{choi2007downlink}. 
	
	An advanced form of Network MIMO, known as cell-free massive MIMO (CF-mMIMO), has garnered considerable attention in recent years \cite{ngo2017cell}. In CF-mMIMO systems, a large number of distributed access points (APs) are connected to a central processing unit (CPU) via backhaul links. This configuration supports collaborative communication, enabling multiple user equipments (UEs) to share the same time-frequency resources. The spatially distributed architecture of CF-mMIMO offers additional macro-diversity gains from dense AP deployment and improves fairness among UEs.
	
	In recent CF-mMIMO literature, the user-centric model has emerged as a highly promising approach. In this model, each UE is served by a subset of APs rather than by all APs \cite{bjornson2020scalable},\cite{lee2025user}. Configuring the AP subset in a user-centric CF-mMIMO network allows for significant scalability improvements, optimizing computational complexity with minimal performance loss compared to the original CF-mMIMO. In user-centric CF-mMIMO, the CPU processes signals from APs based on the configured AP-UE associations, constructing AP subsets as per the chosen association scheme. To maximize network sum spectral efficiency (SE), effective signal processing involves applying appropriate weights to the signals gathered from each AP, thereby enhancing the SE for each UE. This technique, known as large-scale fading decoding (LSFD) \cite{bjornson2019making}, is highly effective. While the user-centric approach reduces computational overhead, CF-mMIMO inherently encounters complex interference among UEs. Consequently, an efficient resource allocation strategy is essential for managing interference and ensuring fair service for a large user base\cite{hao2024joint, tuan2022scalable, you2020energy}. 
	%In this regard, optimal power control is crucial to achieving high performance in CF-mMIMO systems\cite{hao2024joint, tuan2022scalable, you2020energy}.
	
	As network density and the number of UEs continue to increase, energy efficiency (EE) has become as crucial as SE \cite{isheden2012framework, cho2013energy, zou2024energy}. While SE indicates the efficiency of data transmission, EE measures the efficient operation of the entire system. Thus, a primary goal in next-generation wireless communication systems is to maximize EE while maintaining adequate SE. The conventional global EE (GEE) is defined as the ratio of the total data transmission rate to the total power consumption, representing the amount of data transmitted per unit of energy used \cite{zappone2015energy}. In time-division duplex (TDD) wireless communication systems, EE can be evaluated separately for the uplink and downlink phases, each influenced by how the system's power consumption is calculated.
	
	Power consumption can be divided into two main components: network-side power \cite{3gpp38864}, which is used by network elements such as the AP and backhaul, and UE-side power \cite{3gpp38840}, which includes the power needed for transmission, reception, and circuit operations within the device. In the downlink, network-side power consumption is the primary consideration, as the network manages the AP’s transmit power to serve multiple UEs. Conversely, in the uplink, the data transmission rate relies on each UE's transmit power, making UE-side power consumption a critical factor.
	
	Previous studies on the EE of CF-mMIMO systems have mainly focused on downlink global energy efficiency (DL-GEE), defined as the ratio of the total downlink data transmission rate to total network power consumption \cite{papazafeiropoulos2021towards, alonzo2019energy, zheng2023energy, raghunath2024energy}. In contrast, there has been limited research addressing the uplink EE in CF-mMIMO systems. In \cite{chen2023energy}, large-scale fading processing (LSFP) methods were proposed to maximize the GEE for both uplink and downlink. However, these methods relied on heuristic algorithms for power allocation. The work in \cite{bashar2019energy} proposed a method to define uplink global energy efficiency (UL-GEE) when quantization occurs at the APs. It optimized the uplink transmit power using successive convex approximation (SCA) and sub-optimal geometric programming (GP). Another work \cite{choi2021energy} proposed a power control strategy aimed at maximizing max-min fairness in the energy efficiencies of individual UEs, rather than focusing on GEE. The latest study \cite{zhao2024towards} examined the quantization bits and UE transmit power required to optimize UL-GEE in uplink CF-mMIMO with low-resolution ADCs. Since these studies focused on GEE or EE for the UE with the worst performance, addressing the EEs of individual UEs separately remains challenging.
	
	Departing from the conventional GEE, our primary focus is on the EE of individual UEs, defined as the ratio of each UE's SE to its power consumption. Specifically, we utilize the weighted-sum energy efficiency (WSEE) as our key performance metric, incorporating weights that reflect each UE’s unique characteristics, such as power budget and data rate requirements. The WSEE metric is particularly suited for multi-user systems with diverse traffic requirements, as increasingly expected in IMT-2030, because it enables prioritization among UEs through appropriately chosen weights. 
	
	The proposed WSEE is defined as the weighted sum of the EE of individual UEs. This provides flexibility to reflect the unique requirements and priorities of each UE, which cannot be addressed in the overall power consumption related to the total transmission rate. For instance, if a specific UE temporarily has a service requirement with higher reliability or is limited in power usage, a larger weight may be given to that UE to preferentially improve the EE of that UE. Furthermore, it is possible to separate UEs into several groups and, by appropriately adjusting the weights, better serve each group according to its priority while its EE is more strictly managed. WSEE can effectively balance the two goals of optimizing the EE of the entire network and meeting the requirements for each UE. In addition, the proposed WSEE can control the power use of each UE particularly effectively, focusing on the power consumption on the UE-side, thereby extending battery life and contributing to reducing energy costs. This can enhance the reliability and lifetime of the link in important communication situations. The EE of each UE in the uplink WSEE is significantly affected by not only its own transmit power but also the power control of other UEs in the system. To address this, we propose an energy-efficient power control and distributed signal processing framework suitable for uplink user-centric CF-mMIMO systems.
	
	The main contributions of this paper are summarized as below.
	\begin{itemize}
		\item[$\bullet$] Our work is the first to analyze uplink WSEE in CF-mMIMO systems. We propose the new metric of uplink WSEE, which allows for comprehensive evaluations of uplink EE with various priorities of UEs. We also maximize the proposed WSEE by optimizing each UE’s uplink transmit power and the LSFD combining vector at the CPU, once all UEs’ weights and AP-UE associations are determined. 
		\item[$\bullet$] We present a fast optimization algorithm for the WSEE maximization problem. This is crucial as WSEE maximization leads to a non-convex fractional program that is NP-hard. To render it tractable, we employ the quadratic transform (QT), which is well suited for optimizing multiple-ratio fractional objectives. We derive an algorithm that applies QT sequentially to each UE’s EE and to the corresponding SINR. Compared with benchmark schems based on Dinkelbach transform, the proposed approach achieves substantially higher WSEE. Our QT-based alternating optimization (AO) solves, at each iteration, subproblems that are either concave or globally optimizable, and thus guarantees monotone convergence. This yields fast and dependable optimization suitable for practical network operation.
		\item[$\bullet$] Furthermore, we introduce a modified variant that incorporates auxiliary variables in the sequential QT procedure, preserving convergence while accelerating optimization. Both the original and modified methods are based on AO and have comparable per-iteration computational complexity, although their convergence speeds differ depending on the objective and constraints. The modified algorithm improves convergence speed by more than 30\% while incurring less than a 1\% loss in final WSEE performance.
		\item[$\bullet$]  We design the WSEE optimization problem to incorporate each UE's maximum power limit and minimum data rate requirements. Traditional GEE-focused approaches prioritize overall network efficiency, so there are limitations in ensuring the performance of individual UEs because the performance degradation of a particular UE may be offset by the performance improvement of another UE. Our proposed framework—leveraging the WSEE metric and a QT-based solution—enables efficient determination of system operation policies in heterogeneous systems where UEs with diverse characteristics coexist. Simulation results quantify this benefit. In scenarios with diverse UE requirements, the proposed framework improves the uplink WSEE by more than 20\% compared to the GEE-focused baseline, while achieving this gain with a 16\% reduction in runtime.
	\end{itemize}
	
	\subsection{Outline} 
	The remainder of this paper is organized as follows. Section~\ref{sec:system_model} describes the uplink user-centric CF-mMIMO system model. Section~\ref{sec:per_met} defines the uplink EE of individual UEs and the WSEE metric based on the uplink EE. Additionally, we present the weighted-global energy efficiency (WGEE) as a comparative metric for WSEE. In Section~\ref{sec:prob_alg}, we formulate the WGEE and WSEE maximization problems and solve the corresponding optimization problems. To achieve this, we develop an AO algorithm based on QT and an accelerated variant that employs auxiliary variables to mitigate nonlinearity. Section~\ref{sec:num_res} provides a comprehensive performance evaluation of the proposed algorithms under various weight conditions. Finally, Section~\ref{sec:conclusion} concludes the paper.
	
	\subsection{Notations} 
	For a scalar $x$, superscript $x^\dg$ denotes the complex conjugate. For a vector $\mathbf{a}$ and a matrix $\mathbf{A}$, superscript $(\cdot)^{\mathrm{H}}$ denotes the Hermitian. $\Vert \mathbf{a} \Vert$, $\mathrm{tr}(\mathbf{A})$ and $\mathrm{diag}(\mathbf{a})$ are used as the Euclidean norm of $\mathbf{a}$, the trace of $\mathbf{A}$, and the diagonal matrix whose diagonal elements are $\mathbf{a}$, respectively. $\mathbf{I}_N$ and $\mathbf{0}_N$ are identity and zero matrices. For a set $\mathcal{S}$, $\lvert \mathcal{S} \rvert$ is the cardinality of $\mathcal{S}$. The expression $\mathcal{S} \setminus \{s\}$ stands for the rest of the set except for the element $s$. Also, $\mathbb{E}\{\cdot\}$ means the expectation.
	
	\section{System Model\label{sec:system_model}}
	We consider a CF-mMIMO system with $M$ APs, each equipped with $N$ antennas. There are $K$ single antenna UEs, randomly distributed across a wide area.
	To benefit from the favorable propagation and channel hardening effects of CF-mMIMO, we assume that $MN \gg K$. We adopt a user-centric approach \cite{bjornson2020scalable}, where all APs are connected to the CPU via ideal backhaul, but each AP only serves a subset of UEs. The system model is illustrated in Fig. \ref{fig:system_model}. The subset of APs serving UE $k$ is denoted as $\mathcal{M}_k \subseteq \{1, \dots, M\}$, and the subset of UEs served by AP $m$ is denoted as $\mathcal{D}_m \subseteq \{1, \dots, K\}$. TDD is adopted to enable both uplink and downlink transmissions on the same frequency band and to allow for efficient reciprocity-based channel state information (CSI) acquisition. This analysis focuses on uplink channel estimation and uplink data transmission, excluding downlink transmission. We adopt a block fading channel model, where each coherence block has a length of $\tau_c$ (in samples). During each coherence interval, all UEs transmit pilot sequences in the first $\tau_p$ samples, while the remaining $\tau_c - \tau_p$ samples are allocated to uplink data transmission. After uplink channel estimation, the CPU uses the optimized parameters to control the uplink power of each UE. For simplicity, we assume that the UEs instantly know the optimized parameters, meaning time required for downlink transmission of these parameters is negligible.
	\begin{figure}[!t]
		\centering
		\includegraphics[width=0.5\columnwidth]{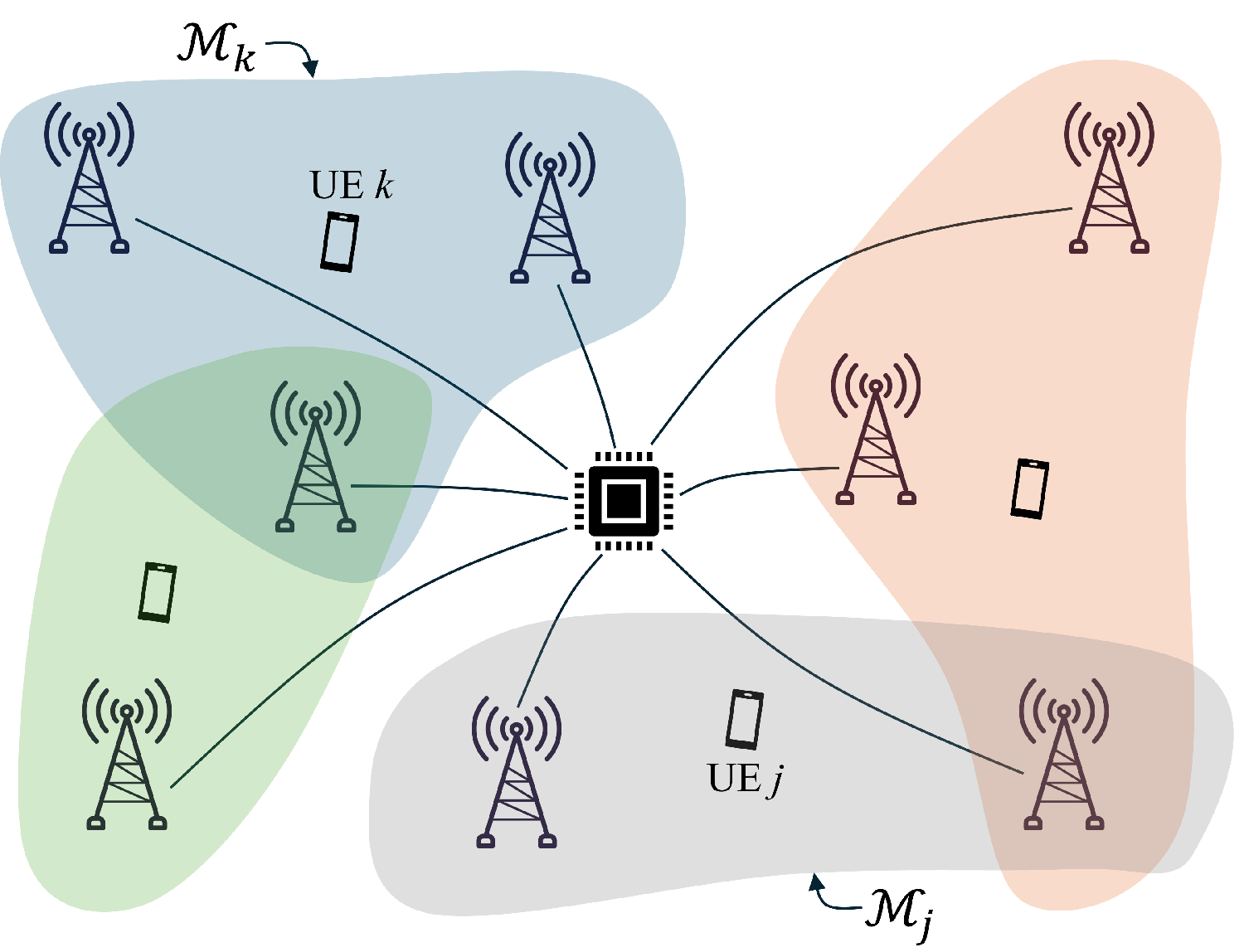}
		\caption{Illustration of a user-centric CF-mMIMO system}
		\label{fig:system_model}
	\end{figure}
	We consider a Rayleigh fading vector channel model, where the channel coefficient vector $\mathbf{g}_{mk} \in \mathbb{C}^{N \times 1}$ between AP $m$ and UE $k$ follows $\mathbf{g}_{mk} \sim \mathcal{CN}(\mathbf{0}, \beta_{mk}\mathbf{I}_N)$. Here, $\beta_{mk}$ represents the large-scale fading coefficient determined by the distance and propagation environment.
	The value of $\beta_{mk}$, representing path loss and shadowing, varies slowly compared to the small-scale fading coefficients. It is assumed that the large-scale fading coefficients $\beta_{mk}$ are known to APs $m \in \{1, \cdots, M\}$. 
	
	\subsection{AP-UE Association}
	Each UE can select its serving AP cluster $\mathcal{M}_k$ based on the large-scale fading information available to it. We use the largest-large-scale-fading-based selection method \cite{papazafeiropoulos2021towards}, where UE $k$ selects the top $\lvert \mathcal{M}_k \rvert$ APs with the best channel conditions. Specifically, the APs in $\mathcal{M}_k$ satisfy the following condition,
	\begin{align} \label{eqn:association}
		\sum_{m=1}^{\lvert \mathcal{M}_k \rvert} \frac{\bar{\beta}_{mk}}{\sum_{m'=1}^{M} \beta_{m'k}} \geq \delta,
	\end{align}
	where ${\bar{\beta}_{1k}, \dots, \bar{\beta}_{Mk}}$ are the sorted large-scale fading coefficients in descending order, and $\delta$ is a parameter between 0 and 1 that indirectly controls the size of the AP cluster serving each UE. Based on the determined association $\mathcal{M}_k$, we define the block-diagonal matrix $\mathbf{D}_k = \mathrm{diag}(\mathbf{D}_{1k}, \dots, \mathbf{D}_{Mk})$, where $\mathbf{D}_{mk} = \mathbf{I}_N$ if $m \in \mathcal{M}_k$, and $\mathbf{D}_{mk} = \mathbf{0}_N$ otherwise.
	
	\subsection{Uplink Channel Estimation}
	During the channel estimation phase, all UEs transmit their pilot sequences, and each AP $m$ locally estimates the channels of UEs in its subset $\mathcal{D}_m$. UE $k$ uses the orthogonal pilot sequence matrix $\mathbf{\Phi} \in \mathbb{C}^{\tau_p \times K}$, where the column $\mathbf{\phi}_k$ is the pilot sequence for UE $k$. If $K \leq \tau_p$, all UEs can use mutually orthogonal pilot sequences. However, in general, $\tau_p \ll K$, leading to pilot contamination. In this scenario, let $\mathcal{P}_k$ represent the set of UEs sharing the same pilot sequence as UE $k$. The signal received by AP $m$ is given by
	\begin{align}\label{eqn:ypilot}
		\mathbf{y}_m^{\mathrm{pilot}} = \sqrt{\tau_p \rho_p} \sum_{k=1}^{K} \mathbf{g}_{mk} \mathbf{\phi}_k^{\mathrm{H}} + \mathbf{n}_m^{\mathrm{pilot}},
	\end{align}
	where $\rho_p$ represents the uplink transmit signal-to-noise ratio (SNR) of the pilot symbols, normalized by the noise power $N_0 = \sigma_{\mathrm{ul}}^2$. The term $\mathbf{n}_m^{\mathrm{pilot}}$ denotes the received noise, with each element independently and identically distributed (i.i.d.) as $\mathcal{CN}(0, 1)$. The minimum mean-squared error (MMSE) estimate of $\mathbf{g}_{mk}$ is given by
	\begin{align}\label{eqn:ghat}
		\mathbf{\hat{g}}_{mk} = \frac{\sqrt{\tau_p \rho_p} \beta_{mk}}{\tau_p \rho_p \sum_{k' \in \mathcal{D}_m} \beta_{mk'} {\lvert\mathbf{\phi}_{k'}^{\mathrm{H}} \mathbf{\phi}_k\rvert}^2 + 1} \mathbf{\tilde{y}}_m^{\mathrm{pilot}},
	\end{align} 
	where
	\begin{align}\label{eqn:ypilottilde}
		\mathbf{\tilde{y}}_m^{\mathrm{pilot}} = \mathbf{y}_m^{\mathrm{pilot}} \mathbf{\phi}_k = \sqrt{\tau_p \rho_p} \mathbf{g}_{mk} + \sqrt{\tau_p \rho_p} \sum_{k' \neq k} \mathbf{g}_{mk'} \mathbf{\phi}_{k'}^{\mathrm{H}} \mathbf{\phi}_k + \mathbf{n}_m^{\mathrm{pilot}} \mathbf{\phi}_k.
	\end{align}
	From the i.i.d. property of small-scale fading and noise, the channel estimate $\mathbf{\hat{g}}_{mk}$ also consists of $N$ i.i.d. Gaussian components.
	For convenience, we define $\gamma_{mk}$ as the mean-square value of a single component of $\mathbf{\hat{g}}_{mk}$, which is given by
	\begin{align}\label{eqn:gamma}
		\gamma_{mk} = \frac{\tau_p \rho_p \beta_{mk}^2}{\tau_p \rho_p \sum_{k' \in \mathcal{P}_k} \beta_{mk'} {\lvert\mathbf{\phi}_{k'}^{\mathrm{H}} \mathbf{\phi}_k\rvert}^2 + 1}.
	\end{align}
	
	\subsection{Uplink Data Transmission}
	In the distributed operation of user-centric CF-mMIMO, signal processing is divided between the AP and the CPU \cite{bjornson2019making}. During the uplink transmission phase, all UEs transmit signals to the APs. The signal received by AP $m$ is given by
	\begin{align}\label{eqn:yul}
		\mathbf{y}_m^{\mathrm{ul}} &= \sum_{k=1}^K \mathbf{g}_{mk} x_k + \mathbf{n}_m,
	\end{align}
	where the transmitted signal $x_k = \sqrt{q_k} s_k$ has a power $q_k$, with $s_k$ being the normalized transmit signal such that $\mathbb{E}\{|s_k|^2\} = 1$. Each AP $m$ uses a local maximum ratio combining (MRC) vector $\mathbf{v}_{mk} = \mathbf{D}_{mk} \mathbf{\hat{g}}_{mk}$, the natural counterpart to conjugate beamforming used in downlink CF-mMIMO literature \cite{bashar2019energy}, \cite{demir2021foundations}. The local estimate of the signal $x_k$ transmitted by UE $k$ at AP $m$ is given by
	\begin{align}\label{eqn:xmkhat}
		\hat{x}_{mk} &= \mathbf{v}_{mk}^{\mathrm{H}} \mathbf{D}_{mk} \mathbf{y}_m^{\mathrm{ul}} = \hat{\mathbf{g}}_{mk}^{\mathrm{H}} \mathbf{D}_{mk} \mathbf{y}_m^{\mathrm{ul}}.
	\end{align}
	After local data estimation at the APs, the local estimates are forwarded to the CPU via backhaul links. The CPU then completes the final data decoding for each UE as follows:
	\begin{align}\label{eqn:xkhat2}
		\hat{x}_k = \sum_{m=1}^{M} u_{mk}^\dg \hat{x}_{mk} = \sum_{m=1}^{M} u_{mk}^\dg  \hat{\mathbf{g}}_{mk}^{\mathrm{H}} \mathbf{D}_{mk} \mathbf{y}_m^{\mathrm{ul}},
	\end{align}
	where $u_{mk} \in \mathbb{C}$ is the weight applied by the CPU to the local estimate $\hat{x}_{mk}$. This process is known as large-scale fading decoding (LSFD). The CPU applies the LSFD matrix $\mathbf{U} \in \mathbb{C}^{M \times K}$, where the $k^{th}$ column $\mathbf{u}_k \in \mathbb{C}^{M \times 1}$ represents the LSFD combining vector for UE $k$. 
	
	Substituting (\ref{eqn:yul}) into the above expression and performing some simple mathematical manipulations, the final estimated result $\hat{x}_k$ can be decomposed into the terms of desired signal $\mathsf{DS}_k$, beamforming uncertainty gain $\mathsf{BU}_k$, inter-user interference $\mathsf{IUI}_{kk'}$, and effective noise $\mathsf{EN}_k$, respectively, as
	\begin{align}
		\hat{x}_k &= \underbrace{\mathbb{E}\left[\sum_{m=1}^M u_{mk}^\dg \hat{\mathbf{g}}_{mk}^{\mathrm{H}} \mathbf{D}_{mk} \mathbf{g}_{mk}\right] x_k}_{\textrm{($\mathsf{DS}_k$, desired signal)}} + \underbrace{\left(\sum_{m=1}^M u_{mk}^\dg \!\left( \hat{\mathbf{g}}_{mk}^{\mathrm{H}} \mathbf{D}_{mk} \mathbf{g}_{mk} \!-\! \mathbb{E}\left\{\hat{\mathbf{g}}_{mk}^{\mathrm{H}} \mathbf{D}_{mk} \mathbf{g}_{mk}\right\}\right) \!\right) x_k}_{\textrm{($\mathsf{BU}_k$, beamforming uncertainty)}} \nonumber \\
		&\quad + \sum_{k' \neq k} \underbrace{\left( \sum_{m=1}^M u_{mk}^\dg \hat{\mathbf{g}}_{mk}^{\mathrm{H}} \mathbf{D}_{mk} \mathbf{g}_{mk'} \right) x_{k'}}_{\textrm{($\mathsf{IUI}_{kk'}$, inter-user interference)}} + \underbrace{\sum_{m=1}^M u_{mk}^\dg \hat{\mathbf{g}}_{mk}^{\mathrm{H}} \mathbf{D}_{mk} \mathbf{n}_m}_{\textrm{($\mathsf{EN}_k$, effective noise)}}.
	\end{align}
	Now, based on the previous discussion, the SE of UE $k$ can be described in closed form, as detailed in the following lemma.
	\begin{lemma}{\cite[Chapter 5]{demir2021foundations}}\label{lemma:sinr}
		In a user-centric CF-mMIMO system with local MRC and LSFD combiner, the uplink SE of UE $k$ can be expressed as 
		\begin{align}\label{eqn:SE}
			\mathrm{SE}_k^{\mathrm{ul}} = \frac{\tau_c - \tau_p}{\tau_c} \log_2 \left( 1 + \mathsf{SINR}_k^{\mathrm{ul}} \right),
		\end{align}
		and $\mathsf{SINR}_k^{\mathrm{ul}}$ is given by
		\begin{align}\label{eqn:SINR_matrix}
			\mathsf{SINR}_k^{\mathrm{ul}} &= \frac{q_k \left|\mathbf{u}_k^{\mathrm{H}} \mathbb{E}\{\hat{\mathbf{g}}_{mk}^{\mathrm{H}} \mathbf{D}_{mk} \mathbf{g}_{mk}\}\right|^2}{\mathbf{u}_k^{\mathrm{H}} \! \left( \! \sum\limits_{k' \neq k} q_{k'} \mathbb{E} \! \left\{ \! \left|\hat{\mathbf{g}}_{mk}^{\mathrm{H}} \mathbf{D}_{mk} \mathbf{g}_{mk'}\right|^2 \! \right\} \! + \! \sigma_{\mathrm{ul}}^2 \mathrm{diag} \! \left( \! \mathbb{E} \! \left\{ \! \left|\hat{\mathbf{g}}_{1k}^{\mathrm{H}} \mathbf{D}_{1k}\right|^2 \! \right\} \!, \! \cdots \! , \mathbb{E} \! \left\{ \! \left|\hat{\mathbf{g}}_{Mk}^{\mathrm{H}} \mathbf{D}_{Mk}\right|^2 \! \right\} \! \right) \! \right) \! \mathbf{u}_k}\\
			&=
			\frac{q_k\left\lvert\sum_{m \in \mathcal{M}_k} u_{m k}^\dg \gamma_{m k}\right\rvert^2}{\sum_{k^{\prime}} q_{k^{\prime}} \sum_{m \in \mathcal{M}_k} \mathfrak{A}_{mk} + \frac{1}{N} \sum_{k^{\prime} \in \mathcal{P}_k \backslash \{k\} } q_{k^{\prime}} \mathfrak{B}_{mk} + \frac{1}{N^2} \sigma_{\mathrm{ul}}^2 \sum_{m \in \mathcal{M}_k} \mathfrak{C}_{mk}},
		\end{align}
		where $\mathfrak{A}_{mk} = \left\lvert u_{m k}\right\rvert^2 \beta_{m k^{\prime}} \gamma_{m k}$, $\mathfrak{B}_{mk} = \left\lvert\sum_{m \in \mathcal{M}_k} u_{m k}^\dg \gamma_{m k} \sqrt{\frac{\rho_{k^{\prime}}}{\rho_k}} \frac{\beta_{m k^{\prime}}}{\beta_{m k}}\right\rvert^2$, and $\mathfrak{C}_{mk} = \left\lvert u_{m k}\right\rvert^2 \gamma_{m k}$.
	\end{lemma}
	
	Note that both the numerator and the denominator of the $\mathsf{SINR}_{k}^{\text{ul}}$ are affine functions with respect to $\mathbf{q}$, which denotes the vector of transmit powers of all UEs.
	In Lemma \ref{lemma:sinr}, given the channel statistics, the SINR of UE $k$ depends on both $\mathbf{u}_k$ and $\mathbf{q}$. When $\mathbf{q}$ is held constant, the SINR reduces to a generalized Rayleigh quotient with respect to $\mathbf{u}_k$. It is also well known that $\mathbf{u}_k$, which maximizes this form of SINR, can be derived in closed form \cite{ghojogh2019eigenvalue}. For our system model, the optimal $\mathbf{u}_k$ is computed as
	\begin{align}\label{eqn:opt_lsfd}
		\mathbf{u}_k = q_k\mathbf{\Gamma}^{-1}\mathbb{E}\{\hat{\mathbf{g}}_{mk}^{\mathrm{H}} \mathbf{D}_{mk} \mathbf{g}_{mk}\},
	\end{align}
	where
	\begin{equation}
		\begin{aligned}
			\mathbf{\Gamma} = \sum_{k' \neq k} q_{k'} \mathbb{E}\left\{\left|\hat{\mathbf{g}}_{mk}^{\mathrm{H}} \mathbf{D}_{mk} \mathbf{g}_{mk'}\right|^2\right\} + \sigma_{\text{ul}}^2 \mathrm{diag}\left(\mathbb{E}\left\{\left|\hat{\mathbf{g}}_{1k}^{\mathrm{H}} \mathbf{D}_{1k}\right|^2\right\}, \dots, \mathbb{E}\left\{\left|\hat{\mathbf{g}}_{Mk}^{\mathrm{H}} \mathbf{D}_{Mk}\right|^2\right\}\right) + \widetilde{\mathbf{D}}_k
		\end{aligned}
	\end{equation} where $\widetilde{\mathbf{D}}_k$ is the matrix obtained by replacing each block of $\mathbf{D}_k$ with its complement matrix.
	
	\section{Performance Metrics\label{sec:per_met}}
	
	\subsection{Weighted-Global Energy Efficiency}
	A straightforward energy-efficiency metric that incorporates per-UE weights is the weighted-global energy efficiency (WGEE), defined as the weighted sum rate (WSR) divided by the total power consumption. WGEE is given by 
	\begin{align}\label{eqn:wgee}
		\mathsf{WGEE} = \frac{\sum_{k=1}^{K} w_k \cdot \mathrm{SE}_k}{\sum_{k=1}^{K} P_k} = B \left(1-\frac{\tau_p}{\tau_c}\right) \cdot \frac{\sum_{k=1}^{K} w_k \log_2 \left( 1 + \mathsf{SINR}_k^{\mathrm{ul}} \right)}{\sum_{k=1}^{K} P_k},
	\end{align}
	where $B$ is the bandwidth, and $P_k$ represents the power consumed by UE $k$ to achieve the corresponding $\mathrm{SE}_k^{\mathrm{ul}}$. This formulation presents WGEE as a single-ratio measure, similar to the conventional GEE. An appropriate weight vector $\mathbf{w}$ can be selected based on the specific characteristics and requirements of both the UEs and the network.

	\subsection{Energy Efficiency of Each User}
	%Before defining the WSEE, 
	We first define the EE of individual UEs as \cite{bjornson2017massive}:
	\begin{align}\label{eqn:EE_k}
		\mathrm{EE}_k = \frac{B \cdot \mathrm{SE}_k^{\mathrm{ul}}}{P_k},
	\end{align}
	In typical DL-GEE analysis, the total network power consumption accounts for the power consumed by distributed APs and backhauls\cite{ngo2017total}. In contrast, in this UL study, we focus solely on UE-side power consumption when defining the uplink WSEE. The power drawn by the APs, backhaul, and CPU is not included in the considered metric, as we are primarily interested in UE-side EE. 
	
	Then, the power consumption model for UE $k$ can be expressed as
	\begin{align}\label{eqn:power_k}
		P_k = P_{\mathrm{UE}_k} = P_{\mathrm{TX}, k} + P_{\mathrm{CP}, k},
	\end{align}
	where $P_{\textrm{TX}, k}$ is the uplink transmit power of the UE, and $P_{\textrm{CP}, k}$ is the circuit power consumed by the UE's hardware. The transmit signal power $P_{\textrm{TX}, k}$ after power amplifier (PA), proportional to the transmit signal power $q_k$, is given by
	\begin{align}\label{eqn:transmit_power}
		P_{\mathrm{TX}, k} = \frac{1}{\zeta_k} q_k,
		%= \frac{1}{\zeta_k} \rho_k N_0 \eta_k,
	\end{align}
	where $\zeta_k$ is the PA efficiency of UE $k$, and $q_k$ is a value having an upper limit of the maximum transmission signal power $p_{\mathrm{max}, k}$. $P_{\mathrm{CP}, k}$ is the power consumed to drive the circuit elements of UE $k$, including cooling, power supply, control signaling, etc. This value is determined by the UE hardware configuration, but may be regarded as a fixed value in an uplink transmission. 
	\begin{remark}
		\textit{The metric ${EE_k}$, defined by (\ref{eqn:EE_k}) $-$ (\ref{eqn:transmit_power}), represents the amount of power $P_k$ that must be consumed by UE $k$ in a user-centric CF-mMIMO system to achieve the desired uplink data rate $B \cdot SE_k^{ul}$. Under the assumptions made in this work, among the factors influencing $P_k$, the only parameter that a UE can directly adjust is its transmit power\footnote{Small-scale LMMSE channel estimation is carried out locally at each AP during the pilot phase; thus no backhaul is involved within the coherence block. The LSFD vectors, computed from large-scale statistics, are pre-shared and updated only when large-scale conditions change. Therefore, slot-level backhaul power is absent or negligible.}. In the case of $SE_k$, it is determined by both the transmit power of each UE and the LSFD combiner at the CPU. The network is responsible for jointly optimizing the transmit power $\mathbf{q}$ and LSFD combiner $\mathbf{U}$. Based on the AP-UE association and the channel information obtained from the APs, the CPU can determine $\mathbf{q}$ and $\mathbf{U}$. Subsequently, downlink control signals can be used to instruct each UE on its transmit power. The UEs then perform uplink data transmission using the power levels received from the network. This operation is ensured by the performance guarantees provided through the TDD block fading channel model.} 
	\end{remark}
	
	\subsection{Weighted-Sum Energy Efficiency }
	We introduce the WSEE as an uplink performance metric, which captures the EE of individual UEs in the previous subsection. The WSEE is given by
	\begin{align}\label{eqn:wsee}
		\mathsf{WSEE} = \sum_{k=1}^{K} w_k \cdot \mathrm{EE}_k = B \left(1-\frac{\tau_p}{\tau_c}\right) \cdot \sum_{k=1}^{K} w_k \cdot \frac{\log_2 \left( 1 + \mathsf{SINR}_k^{\mathrm{ul}} \right)}{P_k},
	\end{align}
	where $w_k$ represents the weight assigned to UE $k$.  In line with WGEE, the weight vector $\mathbf{w}$ can be adjusted to suit the unique characteristics of the UEs and the specific requirements of the network. This flexibility enables the WSEE to be tailored and optimized for a variety of scenarios. As will be shown later in Section \ref{sec:num_res}, this metric allows for comprehensive evaluations of uplink EE under various network scenarios.
	Even when WSEE and WGEE use the same $\mathbf{w}$, the resulting values are not directly comparable. Instead, the EE of individual UEs should be recalculated based on $\mathbf{q}$ and $\mathbf{U}$, which maximize WSEE and WGEE, respectively. The network's performance can then be assessed by comparing these values.

	\section{ Energy Efficiency Maximization \label{sec:prob_alg}}
	\begin{algorithm}[!t]
		\caption{Dinkelbach algorithm for WGEE maximization}
		\label{alg:wgee}
		\begin{algorithmic}[1]
			\State \textbf{Input:} $\mathbf{w}$, tolerance $\varepsilon$ and $\textit{max\_iter}$
			\State \textbf{Initialize} \\
			\quad 1) $\mathbf{q}^{(0)}$ and $\mathbf{U}^{(0)}$ to feasible values \\
			\quad 2) Set $i=0$, $\mathsf{WGEE}^{(0)}=0$
			\Repeat
			\State Update $\mathbf{U}^{(i+1)} \leftarrow \mathbf{U}^{(i)}$ for $\mathbf{q}^{(i)}$
			\State Update $\mathbf{v}$ by (\ref{eqn:vupdate})
			\State Solve $P_{\mathsf{WGEE}}^{\text{Dink}}$ and update $\mathbf{q}^{(i+1)} \leftarrow \mathbf{q}^{(i)}$
			\State Calculate $\mathsf{WGEE}$ \textit{w.r.t} $\mathbf{w}$, $\mathbf{q}$ and $\mathbf{U}$
			\State $i \leftarrow i+1$
			\Until{$\lvert\mathsf{WGEE}^{(i+1)} - \mathsf{WGEE}^{(i)}\rvert < \varepsilon$ or $i = \textit{max\_iter}$}
			\State \textbf{Output:} last updated $\mathbf{q}$, $\mathbf{U}$ and $\mathsf{WGEE}$
		\end{algorithmic}
	\end{algorithm}
	
	\subsection{WGEE Maximization}
	WGEE can be maximized by jointly optimizing $\mathbf{q}$ and $\mathbf{U}$, given each UE's maximum transmit power, data rate requirement, and corresponding weight. The maximization problem of WGEE is as follows.
	\begin{subequations} \label{P-WGEE}
		\begin{align}
			P_{\mathsf{WGEE}} : \underset{\mathbf{q}, \mathbf{U}}{\text{maximize}} \quad & \mathsf{WGEE}(\mathbf{w}, \mathbf{q}, \mathbf{U}) \label{P-WGEEa} \\
			\text{subject to} \quad &
			\Vert \mathbf{u}_k \Vert = 1, \quad \forall k, \label{P-WGEEb} \\
			& \text{SE}_k(\mathbf{q}, \mathbf{u}_k) \ge r_k, \quad \forall k, \label{P-WGEEc} \\
			\qquad\qquad\quad\!\! \ & 0 \leq q_k \leq p_{\text{max}, k}, \quad \forall k. \label{P-WGEEd}
		\end{align}
	\end{subequations}
	$P_{\mathsf{WGEE}}$ can be reformulated by applying the standard Dinkelbach transform \cite{dinkelbach1967nonlinear}, 
	\begin{subequations} \label{P_Dink}
		\begin{align}
			P_{\mathsf{WGEE}}^{\text{Dink}} : \underset{\mathbf{q}, \mathbf{U}, \mathbf{v}}{\text{maximize}} \quad & \sum_{k=1}^{K} w_k \mathrm{SE}_k - v_k \sum_{k=1}^{K} P_k \label{P_Dinka} \\
			\text{subject to} \quad &
			(\ref{P-WGEEb}) - (\ref{P-WGEEd}), \label{PDinkb} \\
			& \mathbf{v} \in \mathbb{C}^K, \label{PDinkc}
		\end{align}
	\end{subequations}
	where the element $v_k$ of the auxiliary vector $\mathbf{v} = [ v_1, v_2, ..., v_K ]$ is updated through
	\begin{align}\label{eqn:vupdate}
		v_k^* = \frac{\sum_{k=1}^{K} w_k \mathrm{SE}_k \left( \mathbf{q}, \mathbf{U} \right) }{\sum_{k=1}^{K} P_k\left( \mathbf{q} \right)}, \quad \forall k.
	\end{align}
	Note that solving the problem  $P_{\mathsf{WGEE}}^{\text{Dink}}$ directly is challenging since the objective is non-concave jointly with respect to $\mathbf{q}$ and $\mathbf{U}$. To handle this, we employ an AO approach, which separates the original problem into two subproblems that alternately optimize $\mathbf{q}$ and $\mathbf{U}$. In the $\mathbf{U}$-subproblem, with $\mathbf{q}$ fixed, the optimal $\mathbf{U}$ is determined using (\ref{eqn:opt_lsfd}). In the $\mathbf{q}$-subproblem, with $\mathbf{U}$ fixed, the optimization problem for $\mathbf{q}$ is solved to maximize the WGEE. The iterative algorithm to solve $P_{\mathsf{WGEE}}^{\text{Dink}}$ is presented in Algorithm \ref{alg:wgee}.

	Referring to (\ref{P_Dinka}), it becomes evident that optimizing WGEE is closely associated with maximizing WSR. Consequently, in the optimization of $P_{\mathsf{WGEE}}^{\text{Dink}}$, alternative approaches, such as the weighted minimum mean square error (WMMSE) method \cite{feng2021weighted}, commonly employed for WSR maximization, or the Lagrangian dual form of QT \cite{shen2018fractional2}, could be considered instead of Algorithm \ref{alg:wgee}. However, directly applying WMMSE or the Lagrangian dual form of QT to our problem poses challenges, as these methods typically do not account for QoS constraints. In Section \ref{sec:num_res}, we use these methods as benchmarks for optimization algorithms in scenarios that do not involve QoS considerations.

	\subsection{WSEE Maximization}
	The objective of this study is to maximize the WSEE by determining the optimal transmit power $\mathbf{q}$ and LSFD combiner $\mathbf{U}$ for each UE, while accounting for maximum transmit power, data rate requirements, and associated weights. This can be formulated as follows.
	\begin{subequations} \label{P1}
		\begin{align}
			P_{\mathsf{WSEE}} : \underset{\mathbf{q}, \mathbf{U}}{\text{maximize}} \quad & \mathsf{WSEE}(\mathbf{w}, \mathbf{q}, \mathbf{U}) \label{P1a} \\
			\text{subject to} \quad &
			\Vert \mathbf{u}_k \Vert = 1, \quad \forall k, \label{P1b} \\
			& \text{SE}_k(\mathbf{q}, \mathbf{u}_k) \ge r_k, \quad \forall k, \label{P1c} \\
			\qquad\qquad\quad\!\! \ & 0 \leq q_k \leq p_{\text{max}, k}, \quad \forall k. \label{P1d}
		\end{align}
	\end{subequations}
	As with $P_{\mathsf{WGEE}}$, $P_{\mathsf{WSEE}}$ is highly non-concave jointly with respect to $\mathbf{q}$ and $\mathbf{U}$. Therefore, we develop an another algorithm based on AO. However, the $\mathbf{q}$-subproblem is still challenging under AO because WSEE is inherently a sum of fractions.
	
	For multiple-ratio objects such as WSEE in $P_{\mathsf{WGEE}}$, Dinkelbach-like transform can be applied by individually transforming each fraction with the Dinkelbach method \cite{rodenas1999extensions}.
	However, applying the Dinkelbach transform individually to each ratio in WSEE, where the power of one UE affects the EE of all UE, can significantly degrade the expected performance. Also, it is known that the Dinkelbach-like transform generally fails to converge even for the simplest forms of multiple-ratio problems.
	To deal with these problems, QT has been proposed to ensure that the transformed problem retains the same objective value as the original \cite{shen2018fractional}. In addition, QT ensures the equivalence to the optimal solution of the original problem.
	
	\begin{algorithm}[!t]
		\caption{Nested-QT based algorithm for WSEE maximization}
		\label{alg:QT2}
		\begin{algorithmic}[1]
			\State \textbf{Input:} $\mathbf{w}$, outer tolerance $\varepsilon_{\text{out}}$ and $\textit{max\_iter}$
			\State \textbf{Initialize} \\
			\quad 1) $\mathbf{q}^{(0)}$ and $\mathbf{U}^{(0)}$ to feasible values \\
			\quad 2) Set $i=0$, $\mathsf{WSEE}^{(0)}=0$
			\Repeat
			\State Solve $\mathbf{U}$-subproblem: $\mathbf{U}^{(i+1)} \leftarrow \mathbf{U}^{(i)}$ for fixed $\mathbf{q}^{(i)}$
			\State Update $\mathbf{z}$ by (\ref{eqn:yupdate})
			\State Update $\mathbf{y}$ by (\ref{eqn:zupdate})
			\State  Solve $\mathbf{q}$-subproblem: $\mathbf{q}^{(i+1)} \leftarrow \mathbf{q}^{(i)}$ for fixed $\mathbf{U}^{(i+1)}$
			\State Return $\mathsf{WSEE}^{(i+1)}$ \textit{w.r.t} $\mathbf{w}$, $\mathbf{q}^{(i+1)}$ and $\mathbf{U}^{(i+1)}$
			\State $i \leftarrow i+1$
			\Until  $\lvert\mathsf{WSEE}^{(i+1)} - \mathsf{WSEE}^{(i)}\rvert < \varepsilon_{\text{out}}$ or $i = \textit{max\_iter}$
			\State \textbf{Output:} last updated $\mathbf{q}$, $\mathbf{U}$ and $\mathsf{WSEE}$
		\end{algorithmic}
	\end{algorithm}
	\begin{algorithm}[!t]
		\caption{Modified $\mathbf{q}$-subproblem}
		\label{alg:qsub}
		\begin{algorithmic}[1]
			\State \textbf{Input:} $\mathbf{y}^{(i+1)}$, $\mathbf{z}^{(i+1)}$, $\mathbf{q}^{(i)}$, inner tolerance $\varepsilon_{in}$, and $\textit{max\_iter}$
			\State \textbf{Initialize} $j = 0$, $\mathbf{q}_{(0)} = \mathbf{q}^{(i)}$
			\Repeat
			\State Calculate $\widehat{\mathsf{SINR}}_k$ by (\ref{PQT3d})
			\State Solve $P_{\mathsf{WSEE}}^{\mathbf{q}\text{-sub}}$ for fixed $\mathbf{z}^{(i+1)}, \mathbf{y}^{(i+1)}$ and  $\widehat{\mathsf{SINR}}_k$
			\State Return $\mathbf{q}_{(j)}$
			\State $j \leftarrow j+1$  
			\Until ${\lVert\mathbf{q}_{(j+1)} - \mathbf{q}_{(j)}\rVert}^2 < \varepsilon_{in}$ or $j = \textit{max\_iter}$
			\State \textbf{Output:} last updated $\mathbf{q}$
		\end{algorithmic}
	\end{algorithm}
	
	In the case of WSEE maximization, every $\mathrm{EE}_k$ is in the form of a single ratio. Also, $\mathrm{SE}_k$ is in the form of $\log_2(1+\mathsf{SINR}_k^{\mathrm{ul}})$, where $\mathsf{SINR}_k^{\mathrm{ul}}$ itself is a single ratio. 
	It is known that QT preserves optimality for multiple-ratio FP problems whose objective is a monotonically increasing function of the individual ratios. Under this condition, the optimal solution to the transformed problem coincides with that of the original problem.
	Since both weighted summation with non-negative weights and $\log_2(1+x)$ are monotonically increasing functions, successive applications of QT to $\mathrm{EE}_k$ and $\mathsf{SINR}_k^{\mathrm{ul}}$ maintain equivalence with the original problem. First, we can reformulate problem  $P_{\mathsf{WSEE}}$ using QT for each $\mathrm{EE}_k$ as follows.
	\begin{subequations} \label{P_QT}
		\begin{align}
			P_{\mathsf{WSEE}}^{\text{QT}} : \underset{\mathbf{q}, \mathbf{U}, \mathbf{y}}{\text{maximize}} \quad & 
			%B \cdot \left(1-\frac{\tau_p}{\tau_c}\right) \cdot 
			\sum_{k=1}^{K} w_k \cdot \mathcal{F}_k (\mathbf{q}, \mathbf{U}, \mathbf{y}) \label{PQTa} \\
			\text{subject to} \quad &
			(\ref{P1b}) - (\ref{P1d}), \label{PQTb} \\
			& \mathbf{y} \in \mathbb{C}^K. \label{PQTc}
		\end{align}
	\end{subequations}
	The function $\mathcal{F}_k (\mathbf{q}, \mathbf{U}, \mathbf{y})$ in the objective of $P_{\mathsf{WSEE}}^{\text{QT}}$ is derived by applying the QT to $\mathrm{EE}_k$. It is expressed as
	\begin{align}\label{eqn:F}
		\mathcal{F}_k (\mathbf{q}, \mathbf{U}, \mathbf{y}) = 2y_k \left( \log_2 \left( 1 + \mathsf{SINR}_k^{\mathrm{ul}}\left( \mathbf{q}, \mathbf{U} \right) \right) \right)^{\frac{1}{2}} - y_k^2 P_k\left( \mathbf{q} \right).
	\end{align}
	The element $y_k$ of the auxiliary vector $\mathbf{y} = [ y_1, y_2, ..., y_K ]$ is used for optimizing $\mathrm{EE}_k$ and is iteratively updated for fixed $\mathbf{q}$ and $\mathbf{U}$, as shown below.
	\begin{align}\label{eqn:yupdate}
		y_k^* = \frac{\sqrt{\log_2 \left( 1 + \mathsf{SINR}_k^{\mathrm{ul}}\left( \mathbf{q}, \mathbf{U} \right) \right)}}{P_k \left( \mathbf{q} \right)}, \quad \forall k.
	\end{align}
	Next, by applying QT once more to the $\mathsf{SINR}_k^{\mathrm{ul}}$ of $\mathcal{F}_k (\mathbf{q}, \mathbf{U}, \mathbf{y})$, we can finally transform the original WSEE maximization problem into a tractable form as follows. 
	\begin{subequations} \label{P_QT2}
		\begin{align}
			P_{\mathsf{WSEE}}^{\text{nested-QT}} : \underset{\mathbf{q}, \mathbf{U}, \mathbf{y}, \mathbf{z}}{\text{maximize}} \quad & 
			%B \cdot \left(1-\frac{\tau_p}{\tau_c}\right) \cdot 
			\sum_{k=1}^{K} w_k \cdot \mathcal{G}_k (\mathbf{q}, \mathbf{U}, \mathbf{y}, \mathbf{z}) \label{PQT2a} \\
			\text{subject to} \quad &
			(\ref{PQTb}), (\ref{PQTc}), \label{PQT2b} \\
			& \mathbf{z} \in \mathbb{C}^K. \label{PQT2c}
		\end{align}
	\end{subequations}
	The function $\mathcal{G}_k (\mathbf{q}, \mathbf{U}, \mathbf{y}, \mathbf{z})$ in the objective of $P_{\mathsf{WSEE}}^{\text{nested-QT}}$ is derived by applying the QT to both $\mathrm{EE}_k$ and $\mathsf{SINR}_k^{\mathrm{ul}}$. It is expressed as 
	\begin{align}\label{eqn:G}
		\mathcal{G}_k (\mathbf{q}, \mathbf{U}, \mathbf{y}, \mathbf{z}) \!=\! 2y_k \! \left( \! \log_2 \! \left( \! 1 \! + \! 2z_k \! \left( \mathsf{SINR}_{\text{num}, k}^{\mathrm{ul}} \left( \mathbf{q}, \mathbf{U} \right) \right)^{\frac{1}{2}} \! - \! z_k^2 \mathsf{SINR}_{\text{denom}, k}^{\mathrm{ul}} \left( \mathbf{q}, \mathbf{U} \right) \! \right) \! \right)^{\frac{1}{2}} \! - \! y_k^2 P_k \left( \mathbf{q} \right).
	\end{align}
	Similar to $\mathbf{y}$, the element $z_k$ of the auxiliary vector $\mathbf{z} = [ z_1, z_2, ..., z_K ]$ is used for optimizing $\mathsf{SINR}_k^{\mathrm{ul}}$, and it is given by
	\begin{align}\label{eqn:zupdate}
		z_k^* = \frac{\sqrt{\mathsf{SINR}_{\text{num}, k}^{\mathrm{ul}}\left( \mathbf{q}, \mathbf{U} \right)}}{\mathsf{SINR}_{\text{denom}, k}^{\mathrm{ul}}\left( \mathbf{q}, \mathbf{U} \right)}, \quad \forall k
	\end{align} where $\mathsf{SINR}_{\text{num}, k}^{\mathrm{ul}}$ and $\mathsf{SINR}_{\text{denom}, k}^{\mathrm{ul}}$ represent the numerator and denominator of $\mathsf{SINR}_{k}^{\mathrm{ul}}$, respectively. This element is iteratively updated using (\ref{eqn:zupdate}) for fixed $\mathbf{q}$ and $\mathbf{U}$.  Note that the objective function of the transformed  $P_{\mathsf{WSEE}}^{\text{nested-QT}}$ is concave with respect to $\mathbf{q}$ as stated in the following proposition.
	\begin{proposition}
		The differentiable function \(\mathcal{G}_k(\mathbf{q}, \mathbf{U}, \mathbf{y}, \mathbf{z})\) is concave with respect to \(\mathbf{q}\) when \(\mathbf{U}\), \(\mathbf{y}\), and \(\mathbf{z}\) are fixed.
	\end{proposition}
	\begin{proof}
		The proof is provided in Appendix A.
	\end{proof}
	
	Now we can apply the AO approach to maximize the WSEE after transforming it into $P_{\mathsf{WSEE}}^{\text{nested-QT}}$ as in \eqref{P_QT2}, by optimizing $\mathbf{q}$ and $\mathbf{U}$ iteratively. First, the optimal LSFD combiner $\mathbf{U}$ is calculated for a given transmit power $\mathbf{q}$. Then, with the updated power and LSFD combiner, the auxiliary variables $\mathbf{y}$ and $\mathbf{z}$ are updated. Based on them, the problem $P_{\mathsf{WSEE}}^{\text{nested-QT}}$ is solved to update the WSEE value. These steps are repeated until the WSEE value converges, yielding the optimized $\mathbf{q}$, $\mathbf{U}$, and the optimized value of WSEE. The procedure described above is summarized in Algorithm \ref{alg:QT2}.
	
	In the $\mathbf{q}$-subproblem, the constraint set is convex, except for (\text{\ref{P1c}}). This non-convex constraint can be reformulated as
	\begin{align}\label{eqn:_const_1}
		\mathsf{SINR}_{k}^{\mathrm{ul}}(\mathbf{q}, \mathbf{U}) \geq 2^{\frac{\tau_c}{\tau_c - \tau_p}r_k} - 1.
	\end{align}
	The numerator and denominator of $\mathsf{SINR}_{k}^{\mathrm{ul}}$ can now be separated and expressed as follows. 
	\begin{align}\label{eqn:relax_const_2}
		\mathsf{SINR}_{\text{num}, k}^{\mathrm{ul}}(\mathbf{q}, \mathbf{U}) - (2^{\frac{\tau_c}{\tau_c - \tau_p}r_k} - 1) \cdot \mathsf{SINR}_{\text{denom}, k}^{\mathrm{ul}}(\mathbf{q}, \mathbf{U}) \geq 0.
	\end{align}
	With this reformulation, the constraint becomes a convex constraint with respect to $\mathbf{q}$ for a fixed $\mathbf{U}$, as both $\text{SINR}_{\text{num}, k}^{\text{ul}}$ and $\text{SINR}_{\text{denom}, k}^{\text{ul}}$ are affine in $\mathbf{q}$. Therefore, replacing the minimum rate constraint in  $P_{\mathsf{WSEE}}^{\text{nested-QT}}$ with (\ref{eqn:relax_const_2}) ensures that the $\mathbf{q}$-subproblem is concave, as guaranteed by Proposition 1.

	Although (\ref{eqn:G}) is concave, it remains a complicated nonlinear function. By introducing additional auxiliary variables, the $\mathbf{q}$-subproblem can be recast into an epigraph form that is easier for CVX solvers to handle.
	\begin{subequations}\label{P_QT3}
		\begin{align}
			P_{\mathsf{WSEE}}^{\mathbf{q}\text{-sub}}: \underset{\mathbf{q},\mathbf{s},\mathbf{t}}{\text{maximize}}\quad
			& \sum_{k=1}^{K} w_k\!\left(2y_k t_k - y_k^2 P_k(\mathbf q_k)\right) \label{PQT3a}\\[2pt]
			\text{subject to}\quad
			& t_k^2 \le s_k, \ \forall k, \label{PQT3b}\\
			& s_k \le \Bigl(1-\frac{\tau_p}{\tau_c}\Bigr)
			\log_{2}\!\bigl(1+\widehat{\mathsf{SINR}}_{k}\bigr),\ \forall k, \label{PQT3c}\\
			& \begin{aligned}
				\widehat{\mathsf{SINR}}_{k}
				&\triangleq 2z_k \sqrt{\mathsf{SINR}^{\text{ul}}_{\text{num},k}}
				\\[-2pt] &\quad
				- z_k^{2}\,\mathsf{SINR}^{\text{ul}}_{\text{denom},k} \ \ge 0,\ \forall k,
			\end{aligned} \label{PQT3d}\\
			& (\ref{PQT2b})\text{--}(\ref{PQT2c}), (\ref{eqn:opt_lsfd}), (\ref{eqn:yupdate}), (\ref{eqn:zupdate}). \label{PQT3e}
		\end{align}
	\end{subequations}
	In the modified problem $P_{\mathsf{WSEE}}^{\mathbf{q}\text{-sub}}$, (\ref{PQT3b}) and (\ref{PQT3c}) are convex, and $\widehat{\mathsf{SINR}}_{k}$ is concave in $\mathbf{q}$. Therefore, for fixed $\mathbf{y}, \mathbf{z}, \mathbf{U}$, we can solve the $\mathbf{q}$-subproblem in place of (\ref{PQT2a}) to obtain the desired update. The nested-QT-based AO algorithm that uses the proposed $\mathbf{q}$-subproblem is presented in Algorithms \ref{alg:QT2} and \ref{alg:qsub}.

	\subsection{Convergence Analysis}
	In the proposed Algorithm \ref{alg:QT2},  the objective obtained through successive applications of QT with power control and LSFD combiner in an AO is guaranteed to converge to a local optimum. To prove this, we can revisit the Proposition 1. From the result of Proposition 1, one can express that $\mathcal{G}_k$ is monotonically non-decreasing at each iteration as follows:
	\begin{align}\label{eqn:mono_inc}
		\begin{aligned}
			\mathcal{G}_k(\mathbf{q}^{(i+1)}, \mathbf{U}^{(i+1)}, \mathbf{y}^{(i+1)}, \mathbf{z}^{(i+1)}) 
			&\overset{\text{(a)}}{\geq} \mathcal{G}_k(\mathbf{q}^{(i)}, \mathbf{U}^{(i+1)}, \mathbf{y}^{(i+1)}, \mathbf{z}^{(i+1)}) \\
			&\overset{\text{(b)}}{\geq} \mathcal{G}_k(\mathbf{q}^{(i)}, \mathbf{U}^{(i+1)}, \mathbf{y}^{(i)}, \mathbf{z}^{(i+1)}) \\
			&\overset{\text{(c)}}{\geq} \mathcal{G}_k(\mathbf{q}^{(i)}, \mathbf{U}^{(i+1)}, \mathbf{y}^{(i)}, \mathbf{z}^{(i)}) \\
			&\overset{\text{(d)}}{\geq} \mathcal{G}_k(\mathbf{q}^{(i)}, \mathbf{U}^{(i)}, \mathbf{y}^{(i)}, \mathbf{z}^{(i)})
		\end{aligned}
	\end{align} where
	(a) is achieved by applying a CVX solver to satisfy the Karush-Kuhn-Tucker (KKT) condition during the update process of $\mathbf{q}$, because $\mathcal{G}_k(\mathbf{q} \mid \mathbf{U}, \mathbf{y}, \mathbf{z})$ and also the $\mathbf{q}$-subproblem are concave in $\mathbf{q}$.  This likewise applies when the $\mathbf{q}$-subproblem is solved by replacing it with $P_{\mathsf{WSEE}}^{\mathbf{q}\text{-sub}}$; (b) and (c) hold due to the concavity property of the auxiliary variable introduced when applying QT \cite{shen2018fractional};  finally, (d) follows from the fact that, with $\mathbf{q}$, $\mathbf{y}$, and $\mathbf{z}$ fixed, updating $\mathbf{U}$ reduces to a generalized eigenvalue problem, and solving it yields $\mathbf{u}_k$ that always maximizes $\mathsf{SINR}_k^{\mathrm{ul}}$. Since $\mathcal{G}_k$ is obtained by applying QT to $\mathsf{SINR}_k^{\mathrm{ul}}$ within $\mathcal{F}_k$, the property of QT that preserves the optimal objective value implies that the $\mathbf{u}_k$ computed via (\ref{eqn:opt_lsfd}) maximizes $\mathsf{SINR}_k^{\mathrm{ul}}$ and therefore achieves the global optimum for both $\mathcal{F}_k$ and $\mathcal{G}_k$.
	
	From (\ref{eqn:mono_inc}), it can be seen that the weighted sum of $\mathcal{G}_k$ is also monotonically non-decreasing at each iteration. Since the objective function value is monotonically non-decreasing and there is an upper bound due to the power constraint, the local convergence of Algorithm \ref{alg:QT2} is guaranteed\cite{boyd2004convex}. 
	On the other hand, the Dinkelbach transform-based algorithm used in Algorithm \ref{alg:wgee} does not theoretically guarantee local convergence even for the simplest form of multiple-ratio FP\cite{zappone2015energy}. The experimental convergence of these algorithms is verified through the results presented in Section \ref{sec:num_res}.
	
	\section{Computational Complexity\label{sec:comp}}
	This section analyzes the computational complexity of the algorithms proposed earlier. Algorithms \ref{alg:wgee} and \ref{alg:QT2} adopt an AO approach, where a generalized eigenvalue problem and concave sub-problems are solved iteratively.
	The key determinant in solving the generalized eigenvalue problem lies in computing the inverse of the $M \times M$ matrix for each of the $K$ UEs, which incurs a complexity of $\mathcal{O}(KM^3)$. The power control sub-problem, derived after applying either the Dinkelbach transform or QT, is solved using the interior point method embedded in a CVX solver \cite{mosek}. Given that the system has $K$ UEs, the computational cost of solving the KKT conditions once is $\mathcal{O}(K^3)$.  Furthermore, the number of iterations required for the CVX solver to achieve an inner tolerance $\varepsilon_{\text{in}}$ is $\mathcal{O}(\sqrt{K}\log(1/\varepsilon_{\text{in}}))$.
	To satisfy an overall accuracy of $\varepsilon_{\text{out}}$, the number of AO iterations is given by $I_{\text{iter}} = \mathcal{O}(\log(1/\varepsilon_{\text{out}}))$. Consequently, the total computational complexity of each algorithm is $\mathcal{O}(I_{\text{iter}}(KM^3 + K^{3.5}\log(1/\varepsilon_{\text{in}})))$\cite{ben2001lectures}.
	
	Although all algorithms share the same computational complexity in terms of big-O notation, their actual runtime differs. Unlike the Dinkelbach transform, which solves a linear problem for the auxiliary variable, QT reformulates the numerator and denominator of the objective function into a quadratic surrogate with respect to the auxiliary variable. As a result, the number of outer iterations $I_{\text{iter}}$ required to solve the power control sub-problem is higher for QT than for the Dinkelbach transform. 
	In fact, the convergence speed of QT is known to be slower than the superlinear convergence of the Dinkelbach transform \cite{shen2018fractional, shen2024accelerating}. However, when the $\mathbf{q}$-subproblem is replaced by $P_{\mathsf{WSEE}}^{\mathbf{q}\text{-sub}}$, the order of computational complexity remains the same, while the modified constraints lead to faster solver processing. Consequently, using Algorithms \ref{alg:QT2} and \ref{alg:qsub} together yields a shorter runtime than Algorithm \ref{alg:wgee} or the Dinkelbach-like transform-based scheme.
	
	\section{Numerical Results\label{sec:num_res}}
	\begin{table}[!t]
		\caption{ Parameter values for simulations}
		\label{table:parameters}
		\centering
		\setlength{\tabcolsep}{12pt} % column padding (two-column friendly)
		\renewcommand{\arraystretch}{0.8} % row height
		\begin{tabular}{ll}
			\toprule
			\textbf{Parameter} & \textbf{Value} \\
			\midrule
			Network size, $D$ & $1\mathrm{km}$ \\
			Bandwidth, $B$ & $20\mathrm{MHz}$ \\
			Carrier frequency, $f_0$ & $1.9\mathrm{GHz}$ \\
			Noise figure & $7\mathrm{dB}$ \\
			AP selection threshold, $\delta$ & $0.99$ \\
			PA efficiency, $\zeta$ & $0.4$ \\
			Max.\ power (high-priority UE), $\rho_{p}^{\text{high}},p_{\max}^{\text{high}}$ & $0.5\mathrm{W}$ \\
			Max.\ power (low-priority UE), $\rho_{p}^{\text{low}},p_{\max}^{\text{low}}$ & $0.2\mathrm{W}$ \\
			Min.\ data rate (high-priority UE), $r_k^{\text{high}}$ & $1\mathrm{bps/Hz}$ \\
			Min.\ data rate (low-priority UE), $r_k^{\text{low}}$ & $0.5\mathrm{bps/Hz}$ \\
			Circuit power of UE, $P_{\mathrm{CP}}$ & $1\mathrm{W}$ \\
			Coherence / pilot interval, $\tau_c,\tau_p$ & $200,20$ \\
			\bottomrule
		\end{tabular}
	\end{table}
	\subsection{Simulation Setup} In this section, we evaluate WSEE of uplink user-centric CF-mMIMO systems. In our simulations, we assume a scenario where 256 APs equipped with 4 antennas each and 16 single-antenna UEs are randomly distributed in a square area of size $D \times D$. The large-scale fading coefficient $\beta_{mk}$ is modeled as
	\begin{align}\label{eqn:beta_simul}
		\beta_{mk} = \text{PL}_{mk} z_{mk},
	\end{align}
	where $\text{PL}_{mk}$ represents the path loss between the $m$-th AP and the $k$-th UE, and $z_{mk}$ denotes the log-normal shadowing with a standard deviation of $\sigma_{\mathrm{sh}} = 8$dB. $\text{PL}_{mk}$ is calculated based on the three-slope model described in \cite{ngo2017total}:
	\begin{align}
		\text{PL}_{mk} = 
		\begin{cases} 
			-L - 35 \log_{10}(d_{mk}), \quad \text{if } d_{mk} > d_1, \\
			-L - 15 \log_{10}(d_1) - 20 \log_{10}(d_{mk}), \quad \text{if } d_0 < d_{mk} \leq d_1, \\
			-L - 15 \log_{10}(d_1) - 20 \log_{10}(d_0), \quad \text{if } d_{mk} \leq d_0,
		\end{cases}
	\end{align}
	where $L=140.7$dB is a constant determined by the communication environment, such as carrier frequency and antenna altitude. $d_{mk}$ is the distance between AP $m$ and UE $k$, and it is assumed to be $d_0=10 \text{m}$ and $d_1=50 \text{m}$.
	
	The CF-mMIMO system consists of two types of UEs with different priorities, each present in equal numbers. UEs with greater maximum transmit power and more demanding data transmission requirements are given higher priority, while UEs with smaller maximum transmit power and less stringent data needs are assigned lower priority. The system can determine each UE's weight based on its priority. We define the ratio of the higher-priority UE's weight $w_{\mathrm{high}}$ to the lower-priority UE's weight $w_{\mathrm{low}}$ as $\omega = \frac{w_{\mathrm{high}}}{w_{\mathrm{low}}}$. Additionally, to ensure a fair comparison between results using different values of $\omega$, the sum of all UEs' weights was normalized to be $K$. The values of the other parameters used in the simulation setup are summarized in Table \ref{table:parameters}.
	
	\subsection{Results and Discussions} 
	For the comparison, we consider the following schemes:
	\begin{itemize}[leftmargin=*, itemsep=2pt]
		\item \textbf{Nested-QT}: a scheme that solves $P_{\mathsf{WSEE}}^{\text{nested-QT}}$ directly using Algorithm \ref{alg:QT2}.
		\item \textbf{Modified-QT}: a scheme that solves $P_{\mathsf{WSEE}}^{\text{nested-QT}}$ and $P_{\mathsf{WSEE}}^{\mathbf{q}\text{-sub}}$ using the proposed Algorithms \ref{alg:QT2}–\ref{alg:qsub}, with modified constraints in the $\mathbf{q}$-subproblem to accelerate convergence.
		\item \textbf{Dinkelbach-like}: a scheme that applies the Dinkelbach transform to each $\mathrm{EE}_k$ in the WSEE objective on a per-UE basis and then performs AO.
		\item \textbf{Dinkelbach-global}: a scheme implemented as Algorithm \ref{alg:wgee} that applies the Dinkelbach transform to the single-ratio WGEE objective followed by AO.
		\item \textbf{Full-power}: a scheme that updates the LSFD via (13) while fixing each UE’s transmit power at its maximum value.
		\item \textbf{Benchmark}: the approach in \cite{zhao2024towards} that uses a closed-form QT update for system EE, adapted here to maximize WGEE under the system model considered in this work.
	\end{itemize}
	
	\begin{figure}[t]
		\centering
		\begin{subfigure}[b]{0.49\textwidth}
			\centering
			\begin{overpic}[width=\linewidth]{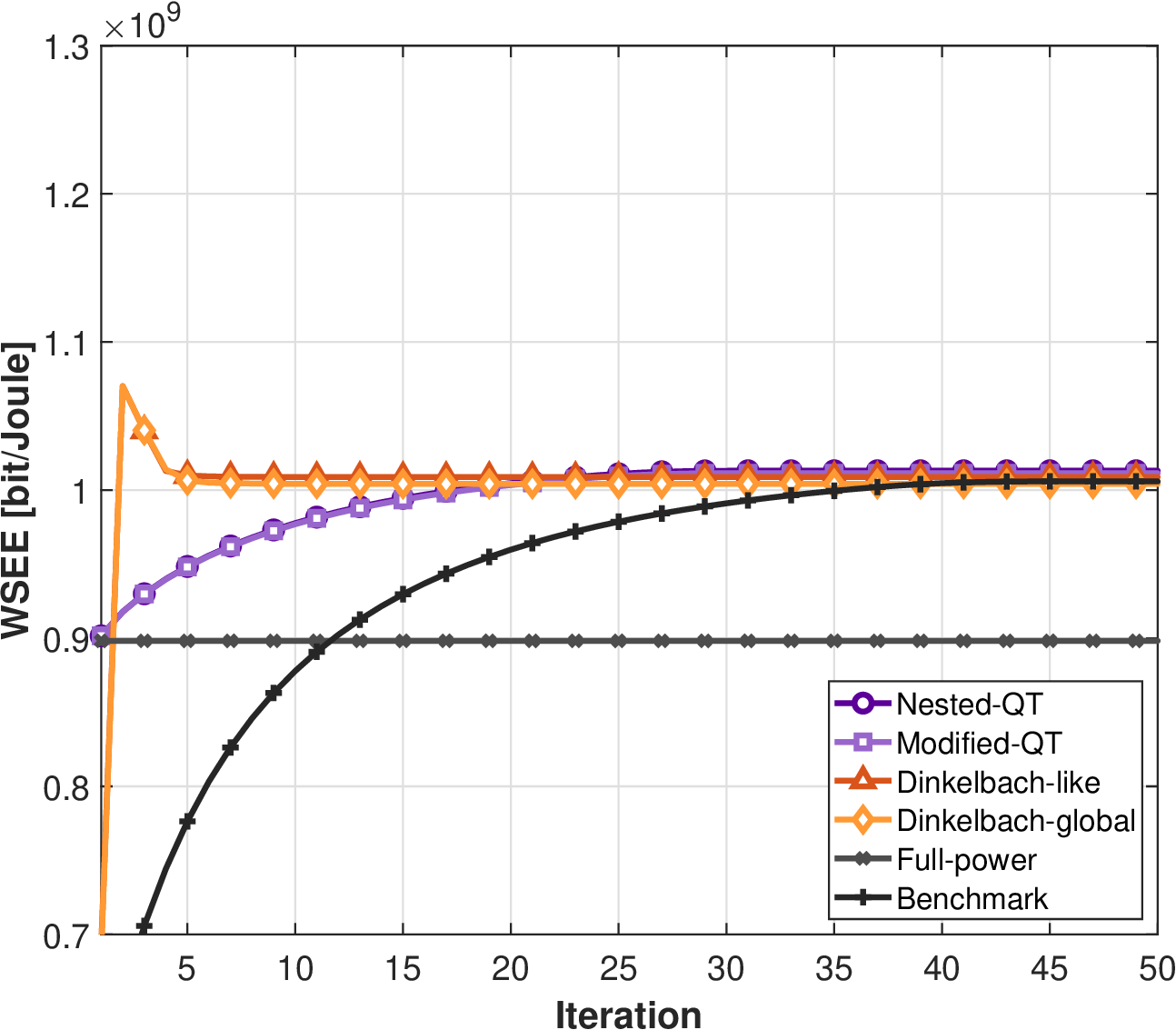}%
				\put(360,550){\includegraphics[width=0.32\linewidth]
					{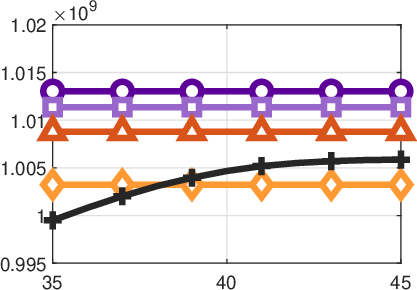}}%
				\put(710,440){\linethickness{0.2mm}\color{black}\polygon(0,0)(185,0)(185,55)(0,55)}%
				\put(670,650){\linethickness{0.2mm}\color{black}\line(1,-2){77}}%
			\end{overpic}
			\caption{When $\omega = 1$.}
			\label{fig:WSEE_iter_omega1}
		\end{subfigure}
		\hfill 
		\begin{subfigure}[b]{0.49\textwidth}
			\centering
			\begin{overpic}[width=\linewidth]{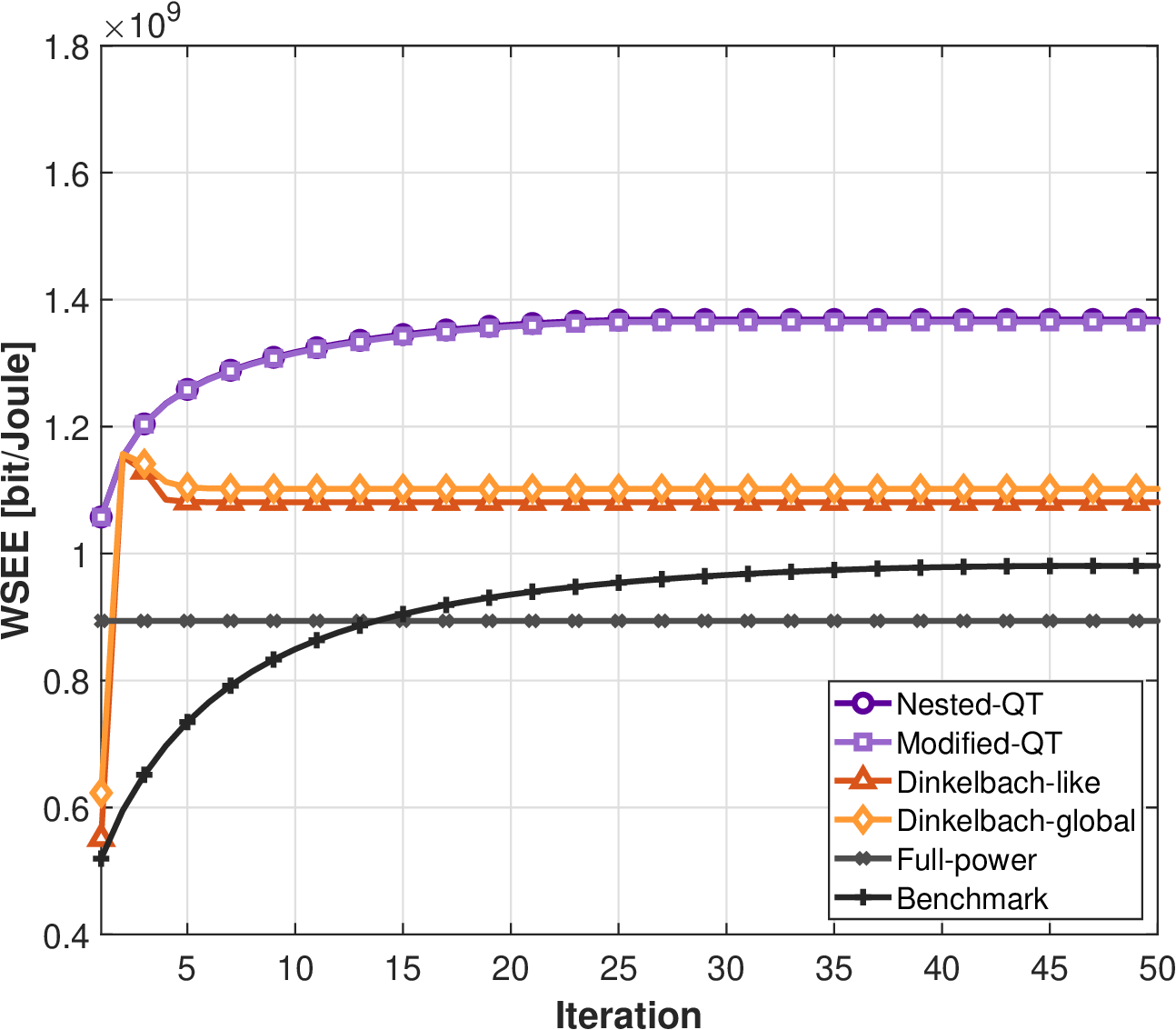}%
				\put(160,650){\includegraphics[width=0.4\linewidth]
					{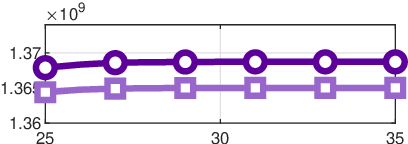}}%
				\put(525,580){\linethickness{0.2mm}\color{black}\polygon(0,0)(185,0)(185,55)(0,55)}%
				\put(480,670){\linethickness{0.2mm}\color{black}\line(1,-1){44}}%
			\end{overpic}
			\caption{When $\omega = 10$.}
			\label{fig:WSEE_iter_omega10}
		\end{subfigure}
		\caption{The WSEE over iterations with different weighting factors.}
		\label{fig:WSEE_iter_combined}
	\end{figure}
	
%	\begin{figure}[!t] 
%		\centering
%		\begin{overpic}[width=0.5\textwidth]{figures/convergence_omega_1.eps}%
%			\put(360,550){\includegraphics[width=0.16\textwidth]
%				{figures/convergence_omega_1_zoom.eps}}%
%			\put(710,440){\linethickness{0.2mm}\color{black}\polygon(0,0)(185,0)(185,55)(0,55)}%
%			\put(670,650){\linethickness{0.2mm}\color{black}\line(1,-2){77}}%
%		\end{overpic}
%		\caption{The WSEE over iterations when $\omega = 1$. \vspace{1\baselineskip}}
%		\label{fig:WSEE_iter_omega1}
%	\end{figure}
%	\begin{figure}[!t] 
%		\centering
%		\begin{overpic}[width=0.5\textwidth]{figures/convergence_omega_10.eps}%
%			\put(160,650){\includegraphics[width=0.2\textwidth]
%				{figures/convergence_omega_10_zoom.eps}}%
%			\put(525,580){\linethickness{0.2mm}\color{black}\polygon(0,0)(185,0)(185,55)(0,55)}%
%			\put(480,670){\linethickness{0.2mm}\color{black}\line(1,-1){44}}%
%		\end{overpic}
%		\caption{The WSEE over iterations when $\omega = 10$. \vspace{1\baselineskip}}
%		\label{fig:WSEE_iter_omega10}
%	\end{figure}
	
	Fig. \ref{fig:WSEE_iter_omega1} illustrates the variation of WSEE with respect to the number of iterations when all UEs are assigned equal weights ($\omega = 1$). In the early iterations, Dinkelbach transform-based algorithms such as \textbf{Dinkelbach-like} and \textbf{Dinkelbach-global} exhibit higher WSEE values. However, in terms of the final performance, the proposed QT-based algorithms, namely \textbf{Nested-QT} and \textbf{Modified-QT}, outperform them. The Dinkelbach transform-based schemes show almost comparable performance to each other, while the QT-based schemes also present nearly identical convergence curves. Nevertheless, the \textbf{Nested-QT} algorithm achieves the best final performance with a slight margin. The Dinkelbach transform-based schemes display a sharp increase in performance at the very beginning followed by a gradual decline, which occurs because, unlike the proposed QT-based algorithms that ensure monotonic convergence, the Dinkelbach transform does not guarantee monotonic convergence. The \textbf{Benchmark} scheme exhibits a similar final performance to the Dinkelbach transform-based schemes, but it requires the largest number of iterations to converge and shows a relatively slower rate of performance improvement.
	
	Fig. \ref{fig:WSEE_iter_omega10} presents the WSEE performance when $\omega = 10$, representing a scenario where higher-priority UEs are ensured greater energy efficiency. Similar to the previous results, the proposed QT-based schemes achieve the best WSEE convergence performance. In terms of the converged value, the \textbf{Nested-QT} scheme outperforms \textbf{Modified-QT} by 0.3\%, \textbf{Dinkelbach-like} by 26.7\%, \textbf{Dinkelbach-global} by 24.2\%, \textbf{Full-power} by 53.1\%, and \textbf{Benchmark} by 39.6\%. Moreover, unlike the case with $\omega = 1$, the QT-based schemes also outperform the Dinkelbach transform-based algorithms in terms of peak value. This implies that the proposed method becomes more advantageous when there exists a difference in the weights among UEs. It also demonstrates that the proposed QT-based algorithms are particularly well-suited for optimizing network performance in CF-mMIMO systems with heterogeneous UE groups, where user priority differentiation is required.
	
	\begin{table}[t]
		\caption{Runtime to convergence across schemes, $M = 256,~ K = 16, ~\omega = 1$.}
		\label{table:conv_time}
		\centering
		\setlength{\tabcolsep}{12pt}
		\renewcommand{\arraystretch}{0.8}
		\begin{tabular}{lcc}
			\toprule
			\textbf{Scheme} & \textbf{Time [s]} & \textbf{Iterations} \\
			\midrule
			Nested-QT         & $0.336$ & 29 \\
			Modified-QT       & $\textbf{0.216}$ & 30 \\
			Dinkelbach-like   & $0.257$ & \textbf{10} \\
			Dinkelbach-global & $5.400$ & 13 \\
			Benchmark         & $2.706$ & 39 \\
			\bottomrule
		\end{tabular}
	\end{table}
	
	\begin{figure}[!t] 
		\centering
		\includegraphics[width=0.6\columnwidth]{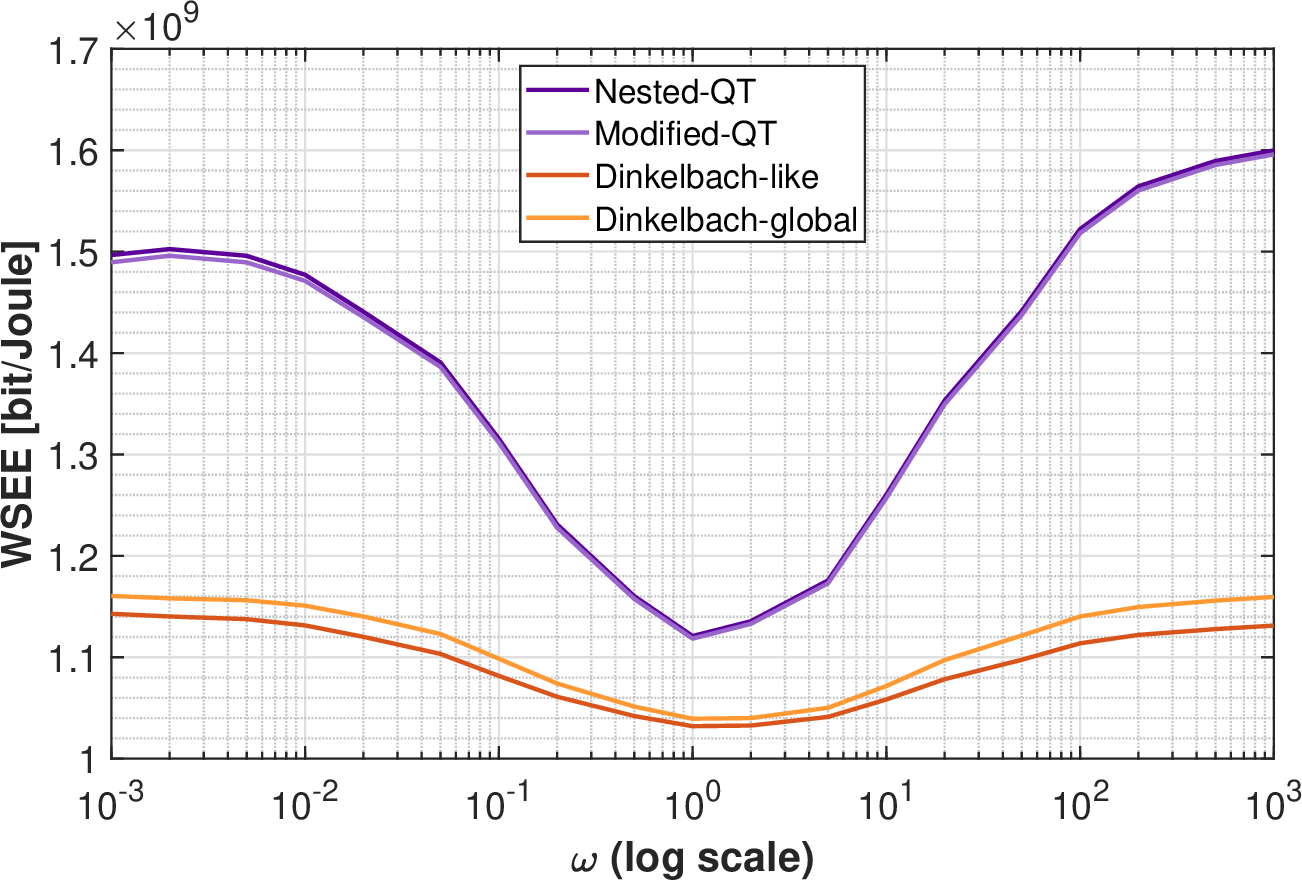}
		\caption{Curves of WSEE versus $\omega$.}
		\label{fig:WSEE_versus_omega}
	\end{figure}
	
	Table \ref{table:conv_time} summarizes the number of iterations and the elapsed time required for each scheme to converge under the same conditions as in Fig. \ref{fig:WSEE_iter_omega10}. In terms of the number of iterations, the Dinkelbach transform-based schemes show the best performance due to their superlinear convergence property, whereas the QT-based schemes generally exhibit a more gradual convergence rate \cite{shen2024accelerating}. From the perspective of total execution time, the proposed \textbf{Nested-QT} scheme is about 31\% slower than \textbf{Dinkelbach-like}. In contrast, the \textbf{Modified-QT} scheme introduces an auxiliary variable to alleviate nonlinearity and maintains nearly the same WSEE performance as \textbf{Nested-QT}. Moreover, it converges approximately 16\% faster than \textbf{Dinkelbach-like}, making it the fastest among all schemes. The \textbf{Dinkelbach-global} and \textbf{Benchmark} schemes require significantly longer convergence times since they employ algorithms designed for GEE maximization, which involves single-ratio optimization representing the entire network system.
	
	\begin{figure}[!t] 
		\centering
		\includegraphics[width=0.6\columnwidth]{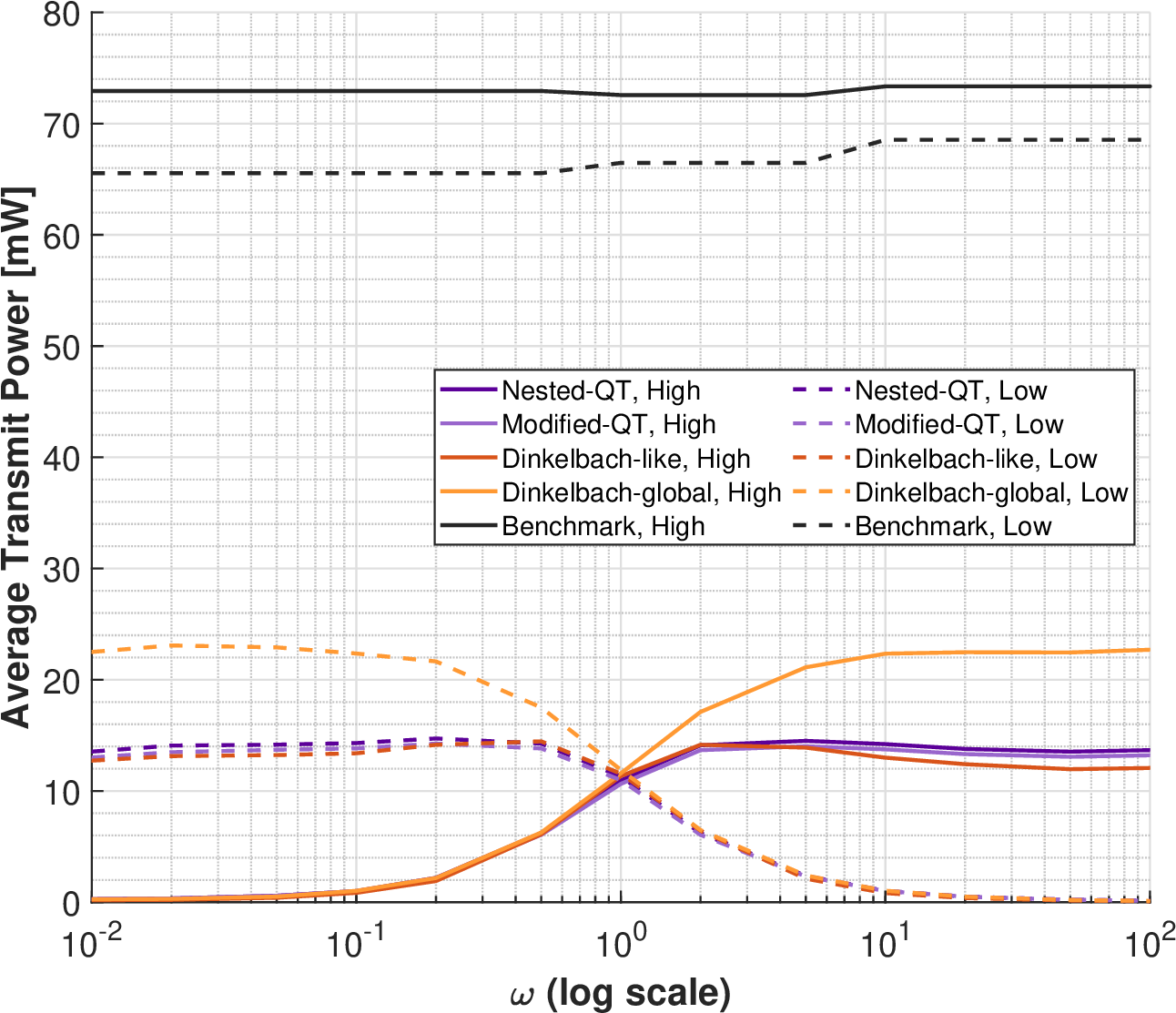}
		\caption{Average transmit power consumption.}
		\label{fig:tpc}
	\end{figure}
	
	Fig. \ref{fig:WSEE_versus_omega} shows the variation of WSEE for the QT- and Dinkelbach transform-based schemes when $\omega$ is varied from 0.001 to 1000. For cases where $\omega < 0.1$ or $\omega > 10$, it is difficult to satisfy the minimum requirement constraints of all UEs, so the optimization is performed without those constraints.
	Across all schemes, the smallest WSEE values are observed at $\omega = 1$, and the converged WSEE values gradually increase as $\omega$ becomes either larger or smaller. This increasing tendency at the extreme values of $\omega$ is more pronounced for the QT-based algorithms. In the QT-based schemes, the change in WSEE with respect to $\omega$ is not symmetric; the increase is slightly greater when $\omega < 1$ than when $\omega > 1$. This is because, for $\omega < 1$, the UEs with smaller SE requirements are given higher priority, leading to larger EE values for high-priority UEs compared to the case when $\omega > 1$. When $\omega$ reaches extreme levels such as 1000 or 0.001, the rate of WSEE increase becomes significantly slower, and the WSEE values are observed to converge.
	
	Fig. \ref{fig:tpc} illustrates the average transmit power variation of each UE group with respect to $\omega$ for different schemes. The QT-based schemes and the Dinkelbach transform-based schemes exhibit similar transmit power distributions for the low-priority UEs. However, for the high-priority UEs, the \textbf{Dinkelbach-like} scheme consumes significantly more transmit power compared to the other schemes. The \textbf{Benchmark} scheme, which updates each UE’s transmit power through a closed-form calculation rather than an optimization process, requires several times higher transmit power than the other schemes. Therefore, from the perspective of uplink transmit power consumption, the proposed \textbf{Nested-QT} and \textbf{Modified-QT} schemes demonstrate the most favorable performance.
	
	\begin{figure}[t]
		\centering
		\begin{subfigure}[b]{0.49\textwidth}
			\centering
			\includegraphics[width=\linewidth]{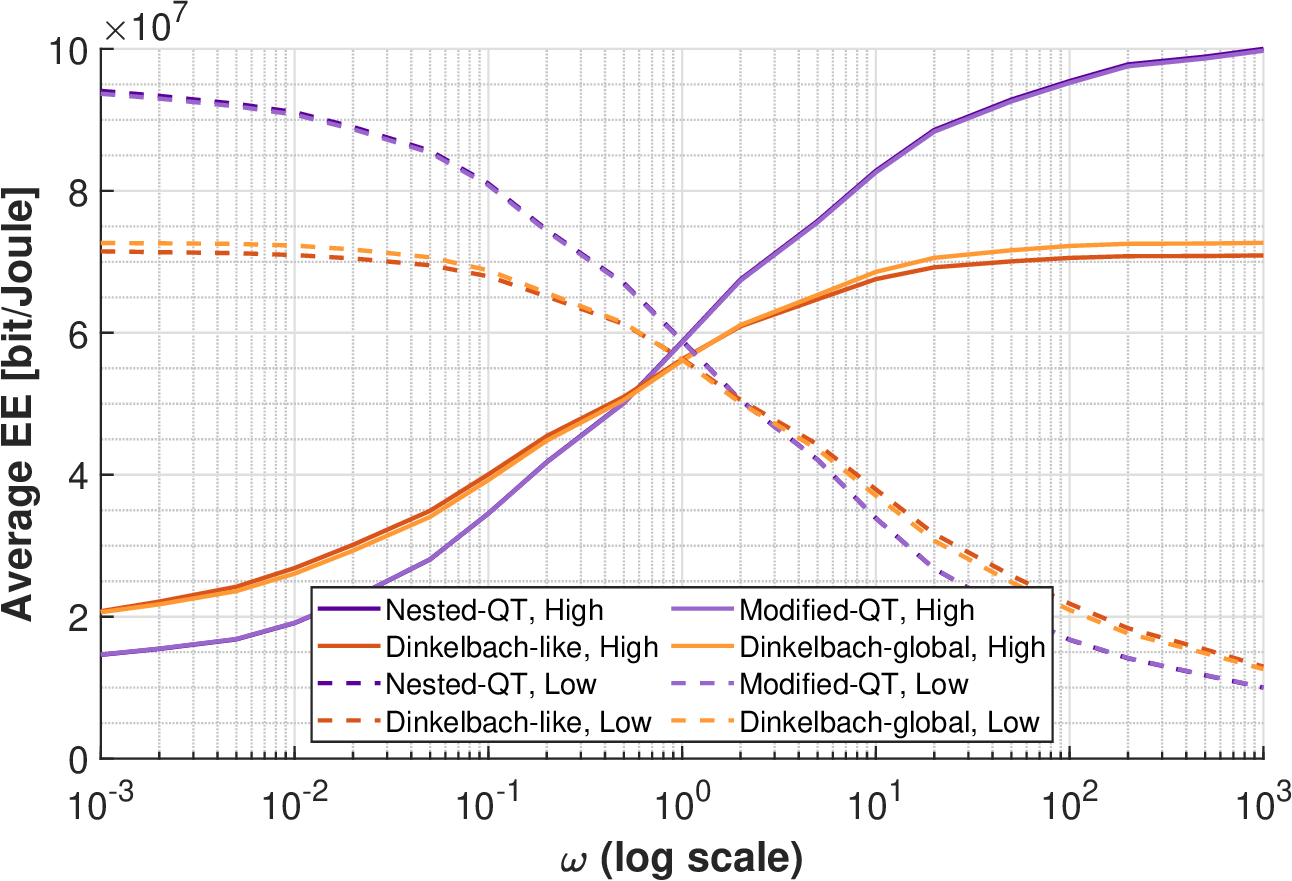}
			\caption{EE distribution.}
			\label{fig:avg_EE}
		\end{subfigure}
		\hfill
		\begin{subfigure}[b]{0.49\textwidth}
			\centering
			\includegraphics[width=\linewidth]{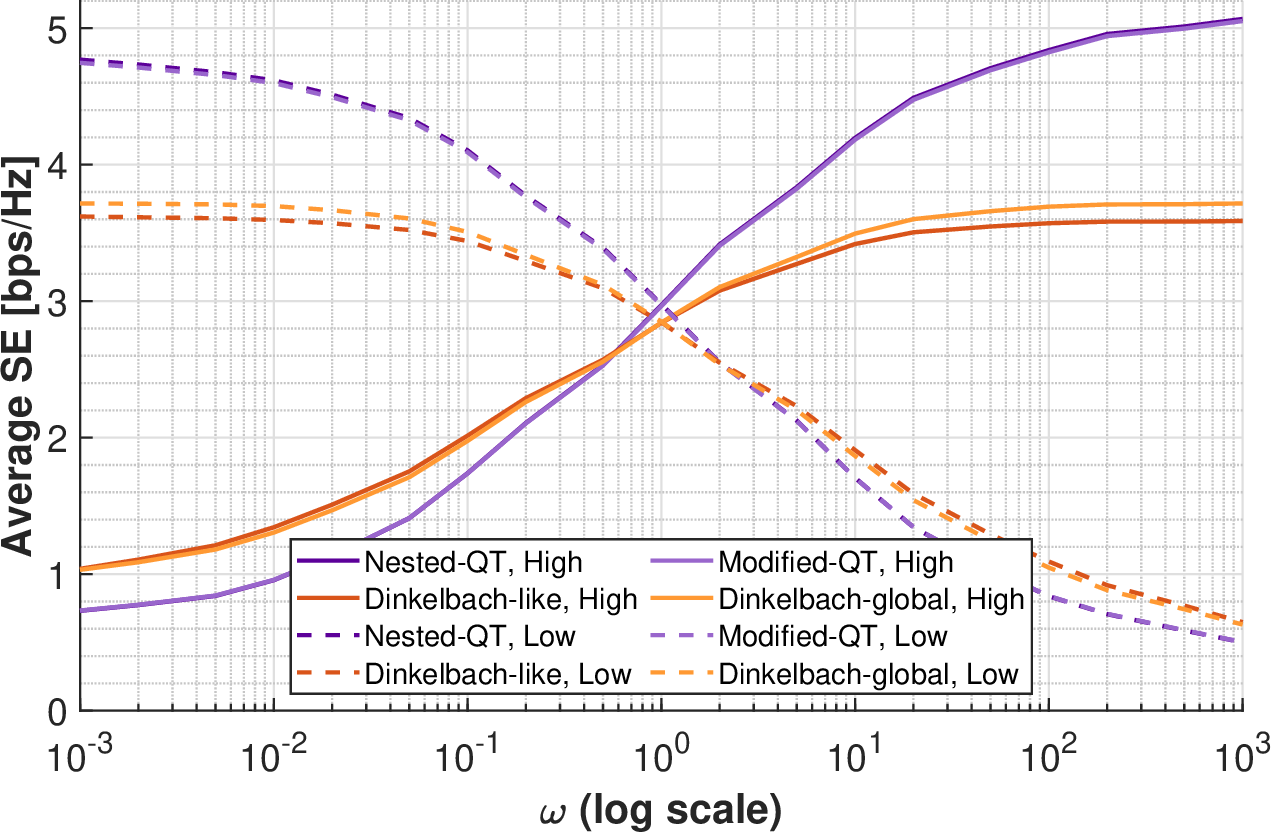}
			\caption{SE distribution.}
			\label{fig:avg_SE}
		\end{subfigure}
		\caption{Average EE and SE distribution by UE priorities over $\omega$.}
		\label{fig:avg_EE_SE_combined}
	\end{figure}
%	\begin{figure}[!t] 
%		\centering
%		\includegraphics[width=0.6\columnwidth]{./figures/Average_EE_versus_omega.eps}
%		\caption{EE distribution by UE priorities over $\omega$.}
%		\label{fig:avg_EE}
%		\vspace{5mm}
%		\centering
%		\includegraphics[width=0.6\columnwidth]{./figures/Average_SE_versus_omega.eps}
%		\caption{SE distribution by UE priorities over $\omega$.}
%		\label{fig:avg_SE}
%	\end{figure}
	Figs. \ref{fig:avg_EE} and \ref{fig:avg_SE} illustrate the variations in EE and SE for each UE group according to their priority under the same conditions as in Fig. \ref{fig:WSEE_versus_omega}. As $\omega$ increases beyond 1, the EE and SE of the high-priority UEs increase, whereas when $\omega$ becomes smaller than 1, the EE and SE of the low-priority UEs increase instead. The EE curve in Fig. \ref{fig:avg_EE} and the SE curve in Fig. \ref{fig:avg_SE} exhibit nearly identical trends, which is because the uplink transmit power of each UE also shows a roughly symmetric behavior with respect to $\omega = 1$. As shown in Fig. \ref{fig:avg_SE}, when $\omega$ becomes extremely large or small, the average SE of certain UEs falls below the minimum requirement of 1 bps/Hz. Therefore, although the overall system WSEE may remain high under such extreme $\omega$ values, communication failures may occur for some UEs. It is thus important to select an appropriate $\omega$ that balances both the overall WSEE performance of the system and the minimum performance requirements of individual UEs.
	
	\begin{figure}[!t] 
		\centering
		\includegraphics[width=0.75\columnwidth]{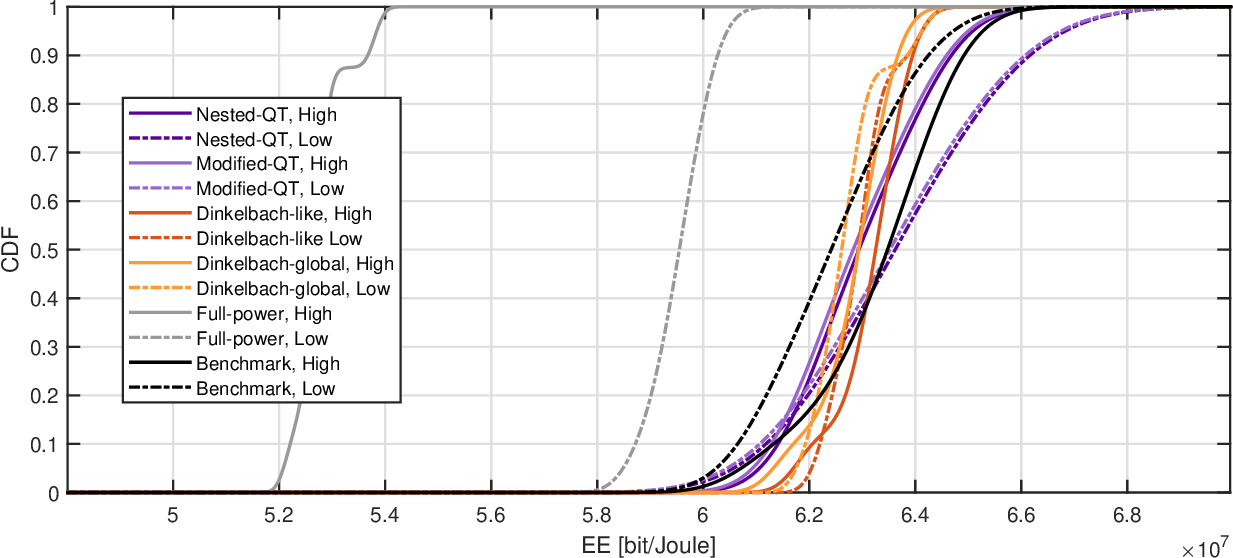}
		\caption{CDF plot of optimization methods and UE priorities, $M = 256,~ K = 16, ~\omega = 1$.}
		\label{fig:cdf_omega1}
		\vspace{5mm}
		\includegraphics[width=0.75\columnwidth]{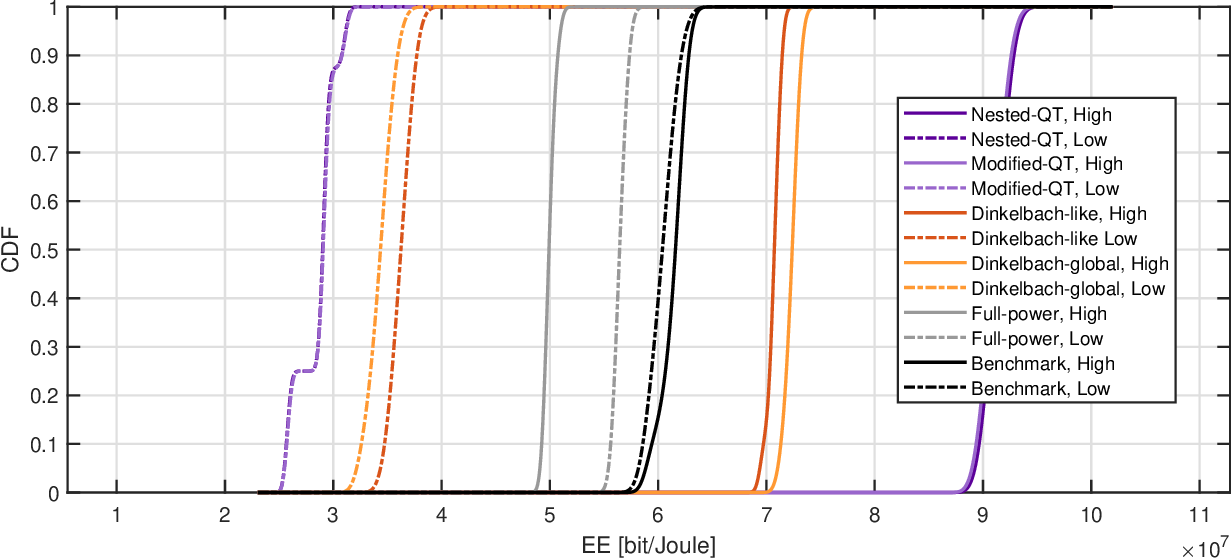}
		\caption{CDF plot of optimization methods and UE priorities, $M = 256,~ K = 16, ~\omega = 10$.}
		\label{fig:cdf_omega10}
	\end{figure}
	Figs. \ref{fig:cdf_omega1} and \ref{fig:cdf_omega10} show the CDF curves of the EE distribution for high- and low-priority UEs in a system with $M = 256$ and $K = 16$, when $\omega$ is set to 1 and 10, respectively. When $\omega = 1$, all schemes except for \textbf{Full-power} exhibit a similar EE range, with most UEs densely concentrated within the same performance band. Among them, however, the proposed \textbf{Nested-QT} and \textbf{Modified-QT} schemes achieve superior performance, particularly for the top 50\% of UEs. When $\omega = 10$, except for the \textbf{Full-power} and \textbf{Benchmark} schemes that are largely unaffected by $\omega$, the optimization-based schemes clearly separate the high-priority and low-priority UE groups. This behavior arises because the weighted optimization algorithm aims to maximize the EE of high-priority UEs while guaranteeing only minimal performance for low-priority UEs. In particular, the proposed \textbf{Nested-QT} and \textbf{Modified-QT} schemes exhibit the most distinct separation between the two groups.
	
	\begin{figure}[!t] 
		\centering
		\includegraphics[width=0.75\columnwidth]{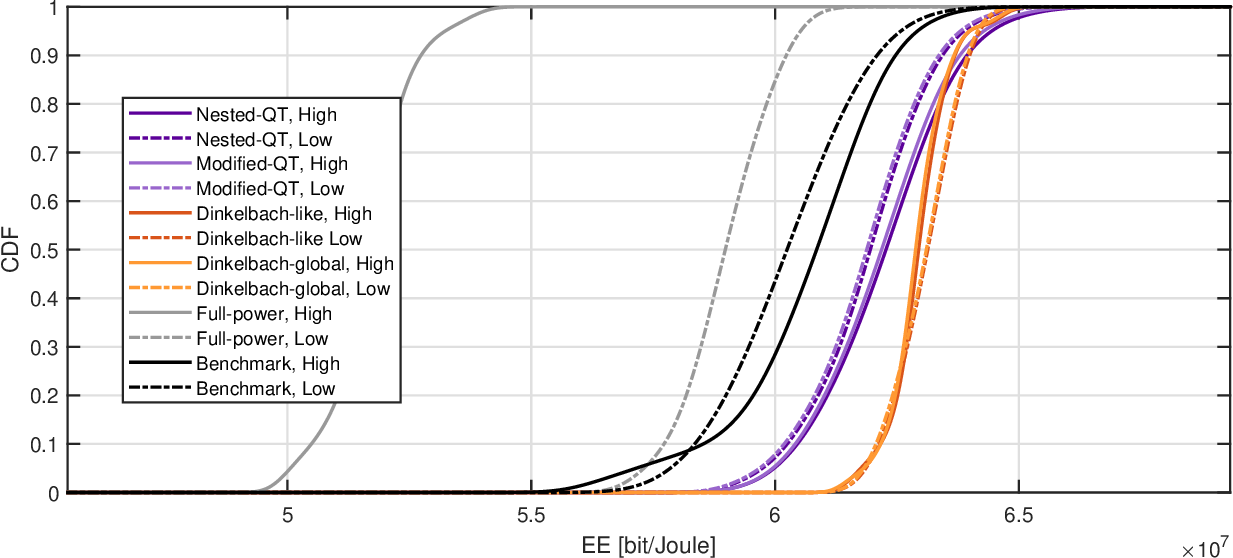}
		\caption{CDF Plot of optimization methods and UE priorities, $M = 1024,~ K = 50, ~\omega = 1$.}
		\label{fig:cdf_omega1_large}
		\vspace{5mm}
		\includegraphics[width=0.75\columnwidth]{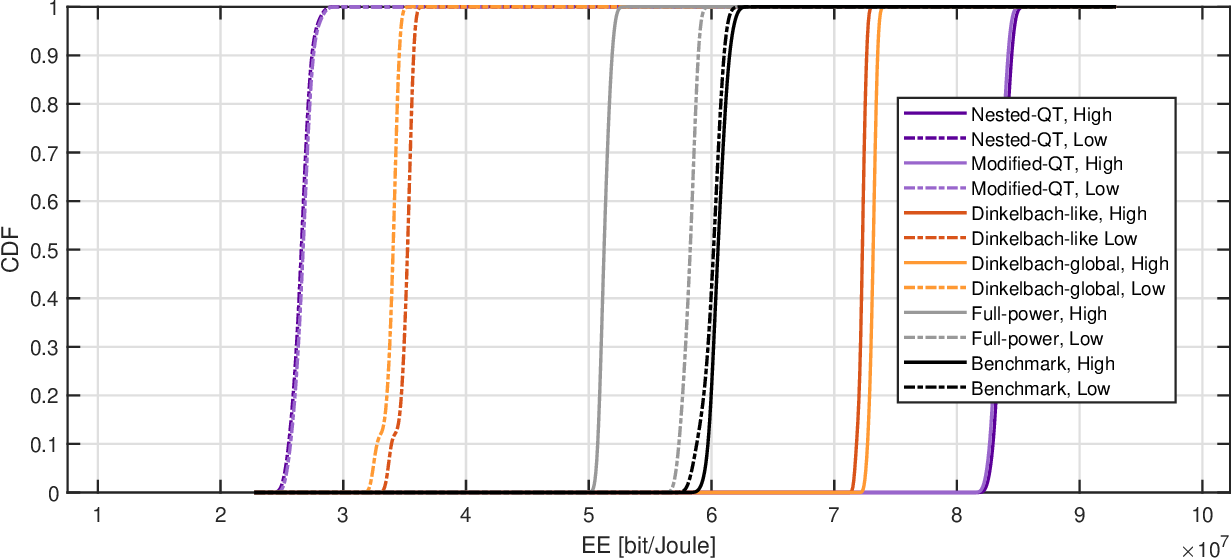}
		\caption{CDF Plot of optimization methods and UE priorities, $M = 1024,~ K = 50, ~\omega = 10$.}
		\label{fig:cdf_omega10_large}
	\end{figure}
	Figs. \ref{fig:cdf_omega1_large}–\ref{fig:cdf_omega10_large} present the CDF curves of the EE distribution for high- and low-priority UEs in a larger network with $M = 1024$ and $K = 50$, under $\omega = 1$ and $\omega = 10$. Consistent with the previous results, when $\omega = 10$, the UEs tend to be separated into distinct groups according to their assigned priorities, whereas this distinction is less apparent when $\omega = 1$. These results confirm that the proposed algorithms maintain the expected behavior even in a large-scale network configuration with 50 simultaneous UEs and a pilot length of $\tau_p = 20$, thereby inevitably leading to pilot contamination.
	
	\section{Conclusions\label{sec:conclusion}}
	We introduced and analyzed the uplink WSEE of a user-centric CF-mMIMO system, expanding beyond the conventional GEE-focused discussions in the literature. Unlike GEE-based analyses, which emphasize maximizing overall energy efficiency rather than the efficiency of individual UEs, WSEE enables an operational strategy tailored to the system's requirements, explicitly accounting for the characteristics and priorities of each UE. To effectively maximize WSEE’s inherently multiple-ratio structure, we developed QT-based optimization algorithms.
	Our numerical results show that, in scenarios with strongly heterogeneous UE requirements and weights, the proposed design can achieve more than 20\% gains in energy efficiency compared to the GEE-focused reference solution, while still outperforming it even in equal service request settings with uniform weights. In addition, the QT-based algorithms provide these gains with substantially reduced runtime and faster convergence compared to existing GEE-focused approaches, making them attractive for practical implementations.
	
	Our framework thus offers a foundation for making system operation policy decisions that can achieve Pareto optimality, especially in systems with UEs of diverse capabilities and requirements. For future work, we plan to enhance realism by incorporating nonlinear PA distortion at the user side and to jointly optimize AP–UE association, thereby improving scalability to much larger network deployments.
	
	\appendices
	\section{Proof of Proposition 1}\label{appi}
	For notational simplicity, we drop the subscript $k$ and represent $\mathcal{G}(\mathbf{q}) \triangleq \mathcal{G}_k(\mathbf{q} \mid \mathbf{U}, \mathbf{y}, \mathbf{z})$ in the proof. From (\ref{eqn:G}), $\mathcal{G}(\mathbf{q})$ can be reformulated as 
	\begin{align}\label{eqn:G2}
		\mathcal{G}(\mathbf{q}) = 2y (\mathcal{H} \circ h) (\mathbf{q}) - y^2P_k(\mathbf{q}).
	\end{align}
	Here, the outer function $\mathcal{H}(x)$ and the inner function $h(\mathbf{q})$ are defined, respectively,  as follows.
	\begin{align}\label{eqn:H}
		\mathcal{H}(x) = \sqrt{\log_2{\left(1+x\right)}},
	\end{align}
	\begin{align}\label{eqn:h}
		h(\mathbf{q}) = 1 + 2z\sqrt{\mathsf{SINR}_{\text{num}, k}^{\mathrm{ul}}(\mathbf{q})} - z^2 \cdot \mathsf{SINR}_{\text{denom}, k}^{\mathrm{ul}}(\mathbf{q}).
	\end{align}
	Since $P_k(\mathbf{q})$ is an affine function of $\mathbf{q}$, as shown in (\ref{eqn:power_k}) and (\ref{eqn:transmit_power}), $\mathcal{G} (\mathbf{q})$ is concave with respect to $\mathbf{q}$ if $(\mathcal{H} \circ h) (\mathbf{q})$ is concave with respect to $\mathbf{q}$. 
	
	To analyze the concavity of $\mathcal{G}(\mathbf{q})$, we represent its second derivative as  
	\begin{align}\label{eqn:second_deriv_G}
		\mathcal{G}''(\mathbf{q}) \propto (\mathcal{H}'' \circ h) (\mathbf{q}) \cdot (h'(\mathbf{q}))^2 + (\mathcal{H}' \circ h) (\mathbf{q}) \cdot h''(\mathbf{q}).
	\end{align}
	It can be easily verified that the first and second derivatives of $\mathcal{H}(x)$ are always positive and negative, respectively, in its domain $x > 0$.
	\begin{align}\label{eqn:first_deriv_H}
		\mathcal{H}'(x) = \frac{1}{2(1+x)\sqrt{\log_2{(1+x)}}} > 0,
	\end{align}
	\begin{align}\label{eqn:second_deriv_H}
		\mathcal{H}''(x) = -\frac{2\log_2{(1+x)} + 1}{4(1+x)^2(\log_2{(1+x)})^{3/2}} < 0.
	\end{align}
	In this paper, the input of $\mathcal{H}(x)$ represents the SINR, which practically takes positive values. For $h(\mathbf{q})$, we have already confirmed that both $\mathsf{SINR}_{\text{num}, k}^{\mathrm{ul}}(\mathbf{q})$ and $\mathsf{SINR}_{\text{denom}, k}^{\mathrm{ul}}(\mathbf{q})$ are affine functions of $\mathbf{q}$. Therefore, it can be identified that if $\sqrt{\mathsf{SINR}_{\text{num}, k}^{\mathrm{ul}}(\mathbf{q})}$ is concave with respect to $\mathbf{q}$, then the entire $h(\mathbf{q})$ is also concave with respect to $\mathbf{q}$. The concavity of $\sqrt{\mathsf{SINR}_{\text{num}, k}^{\mathrm{ul}}(\mathbf{q})}$ can be proven through the following proposition.
	\begin{proposition}
		The composition of a concave function $\phi : \mathbb{R}^{d_2} \to \mathbb{R}$ and an affine function $f : \mathbb{R}^{d_1} \to \mathbb{R}^{d_2}$ is still a concave function.
	\end{proposition}
	\begin{proof}
		See Appendix B.
	\end{proof}
	Since the square root function $\sqrt{x}$ is concave for $x > 0$ and $\mathsf{SINR}_{\text{num}, k}^{\mathrm{ul}}(\mathbf{q}) > 0$ is affine, it follows that the composite function $\sqrt{\mathsf{SINR}_{\text{num}, k}^{\mathrm{ul}}(\mathbf{q})}$ is concave with respect to $\mathbf{q}$. Thus, $h(\mathbf{q})$ is concave and immediately, $h''(\mathbf{q}) < 0$.
	Returning to (\ref{eqn:second_deriv_G}), the first term on the right-hand side is negative since it multiplies a squared term by $\mathcal{H}'' < 0$. The second term is also negative, as it results from the product of $\mathcal{H}' > 0$ and $h''(\mathbf{q})$, which is always negative. Therefore, $\mathcal{G}''(\mathbf{q}) < 0$, proving that $\mathcal{G}(\mathbf{q})$ is concave with respect to $\mathbf{q}$.
	
	\section{Proof of Proposition 2}\label{appii}
	By the nature of the affine function, for any $x_1, x_2 \in \mathbb{R}^{d_1}$ and $\lambda \in [0,1]$, the following holds.
	\begin{align}\label{eqn:affine}
		f(\lambda x_1 + (1 - \lambda) x_2) = \lambda f(x_1) + (1 - \lambda) f(x_2).
	\end{align}
	Now, applying the concave function $\phi$ to (\ref{eqn:affine}) and using the definition of concave functions, we obtain
	\begin{equation}\label{eqn:concave}
		\begin{aligned}
			\phi (f(\lambda x_1 + (1 - \lambda) x_2)) = \phi (\lambda f(x_1) + (1 - \lambda) f(x_2)) \geq \lambda \phi(f(x_1)) + (1 - \lambda) \phi(f(x_2)).
		\end{aligned}
	\end{equation}
	Rewriting the above equation gives
	\begin{equation}\label{eqn:composite}
		\begin{aligned}
			(\phi \circ f)(\lambda x_1 + (1 - \lambda) x_2) 
			&\geq \lambda (\phi \circ f)(x_1) + (1 - \lambda) (\phi \circ f)(x_2).
		\end{aligned}
	\end{equation}
	By the definition of concave functions, it completes the proof that the composite function $(\phi \circ f) (x)$ is concave for x.

	\vspace{-0.1cm}
	\bibliographystyle{IEEEtran}
	\bibliography{[20251216]UC_CF_UL_WSEE}
	
\end{document}